\documentclass[letterpaper,twocolumn,10pt]{article}
\usepackage{usenix}

\usepackage{algorithm, algpseudocode}
\usepackage{soul}

\usepackage{epsfig, endnotes, url, color, graphicx}
\usepackage{amsmath, amsthm, amsfonts}
\usepackage{mathtools, tikz}
\usepackage{diagbox, booktabs, colortbl}
\usepackage{hyperref,listings, framed}
\usepackage{multirow, caption, subcaption}
\usepackage{xspace, slashbox, enumitem, bm, dtrt}
\usepackage[title]{appendix}
\DeclareMathAlphabet{\mathcal}{OMS}{cmsy}{m}{n}
\usepackage{bbding}
\usepackage{wrapfig,lipsum,booktabs}

\newcommand\sysname{$\mathsf{AnoFel}$\xspace}

\newcommand\comm{\mathsf{comm}}
\newcommand\mpk{\mathsf{mpk}}
\newcommand\pk{\mathsf{pk}}
\newcommand\msk{\mathsf{msk}}
\newcommand\sk{\mathsf{sk}}
\newcommand\salt{\mathsf{salt}}
\newcommand\dt{\mathsf{dt}}

\newcommand\sid{\mathsf{sid}}
\newcommand\cl{\mathsf{cl}}
\newcommand\ppt{\mathsf{PPT}\xspace}
\newcommand\prf{\mathsf{PRF}\xspace}
\newcommand\taag{\mathsf{tag}\xspace}

\newcommand\negl{\mathsf{negl}\xspace}
\newcommand\hybrid{\mathsf{Hybrid}\xspace}
\newcommand\error{\mathsf{Er}\xspace}

\newcommand\anongame{\mathsf{AnonGame}\xspace}
\newcommand\dindgame{\mathsf{DINDGame}\xspace}
\newcommand\pafl{\mathcal{O}_{\mathsf{PAFL}}\xspace}
\newcommand\corrupt{\mathsf{corrupt}}
\newcommand\train{\mathsf{train}}
\newcommand\access{\mathsf{access}}
\newcommand\register{\mathsf{register}}
\newcommand\aux{\mathsf{aux}}
\newcommand\stt{\mathsf{state}}
\newcommand\setup{\mathsf{setup}}

\newcommand*\concat{\mathbin{\|}}
\newcommand{\zo}{\{0,1\}}
\newcommand\adv{\mathcal{A}}

\newcommand\PK{\mathcal{PK}}
\newcommand\clients{\mathcal{C}}
\newcommand\agg{\mathcal{AG}}
\newcommand\CM{\mathcal{CM}}

\newtheorem{remark}{Remark}
\newtheorem{theorem}{Theorem}
\newtheorem{lemma}{Lemma}
\newtheorem{definition}{Definition}

\begin{document}

\title{\Large \bf AnoFel: Supporting Anonymity for Privacy-Preserving Federated Learning}

\author{
{\rm Ghada Almashaqbeh}\\
University of Connecticut\\
ghada@uconn.edu
\and
{\rm Zahra Ghodsi}\\
Purdue University \\
zahra@purdue.edu
} %

\maketitle

\begin{abstract}
Federated learning enables users to collaboratively train a machine learning model over their private datasets. Secure aggregation protocols are employed to mitigate information leakage about the local datasets. This setup, however, still leaks the \emph{participation} of a user in a training iteration, which can also be sensitive. Protecting user anonymity is even more challenging in dynamic environments where users may (re)join or leave the training process at any point of time.

In this paper, we introduce \sysname, the first framework to support private and anonymous dynamic participation in  federated learning. \sysname leverages several cryptographic primitives, the concept of anonymity sets, differential privacy, and a public bulletin board to support anonymous user registration, as well as unlinkable and confidential model updates submission. Additionally, our system allows dynamic participation, where users can join or leave at any time, without needing any recovery protocol or interaction. To assess security, we formalize a notion for privacy and anonymity in federated learning, and formally prove that \sysname satisfies this notion. To the best of our knowledge, our system is the first solution with provable anonymity guarantees. To assess efficiency, we provide a concrete implementation of \sysname, and conduct experiments showing its ability to support learning applications scaling to a large number of clients. For an MNIST classification task with 512 clients, the client setup takes less than 3 sec, and a training iteration can be finished in 3.2 sec.
We also compare our system with prior work and demonstrate its practicality for contemporary learning tasks.
\end{abstract}

\section{Introduction}
\label{intro}
Privacy-preserving machine learning is a critical problem that has received huge interest from both academic and industrial communities. Many crucial applications involve training ML models over highly sensitive user data, such as medical screening~\cite{giger2018machine}, credit risk assessment~\cite{galindo2000credit}, or autonomous vehicles~\cite{chen2015deepdriving}. Enabling such applications requires machine learning frameworks that preserve the privacy of users' datasets.

Federated learning (FL) aims to achieve this goal by offering a decentralized paradigm for model training. Participants, or clients, train the model locally over their datasets, and then share only the local gradients with the model owner, or the server. After aggregating updates from all clients, the server shares the updated model with these clients to start a new training iteration. This iterative process continues until the model converges. 

However, individual model updates leak information about clients' private datasets~\cite{melis2019exploiting,nasr2019comprehensive}, and therefore aggregation should be done in a secure way: a server only sees the aggregated value rather than individual contributions from each client. A large body of work emerged to build cryptographic protocols for secure aggregation to support private federated learning, e.g.,~\cite{bonawitz2017practical,bell2020secure,truex2019hybrid,so2021turbo,yang2021lightsecagg,ryffel2020ariann}. Even with a provably secure aggregation protocol, the aggregated model updates still impose a leakage; it was shown that membership inference attacks can determine whether a data sample has been used in the training of a given ML model~\cite{shokri2017membership,nasr2019comprehensive}. Several defense techniques have been proposed that rely on, e.g., regularization techniques to reduce overfitting~\cite{kaya2020effectiveness,wang2021improving}, knowledge distillation~\cite{shejwalkar2019membership}, and differential privacy~\cite{yeom2018privacy}.

\textbf{Anonymous client participation.} A related question to protecting data privacy in federated learning is protecting client identity and breaking linkability with any information that could be deduced from training. Anonymity is critical for training models over sensitive data related to, e.g., rare diseases or sexual abuse incidents. The mere knowledge that a user has participated implies being ill or a victim. It may also allow collecting sensitive information, e.g., location in autonomous vehicles related applications, or financial standing in trading or loans related training tasks. Such leakage invades privacy, and may discourage participation.

Unfortunately, existing frameworks for private federated learning either don't support client anonymity, or suffer from security issues. Secure aggregation protocols~\cite{bonawitz2017practical,bell2020secure} require full identification of clients through a public key infrastructure (PKI) or a trusted client registration process to prevent Sybil attacks and impersonation. Even frameworks that deal with honest-but-curious adversaries~\cite{bonawitz2017practical,truex2019hybrid,so2021turbo,yang2021lightsecagg,ryffel2020ariann} assign clients logical identities, where the mapping between the logical and real identities is known to the server or a trusted third party. On the other hand, techniques that anonymize datasets~\cite{yeom2018privacy,sweeney2002k,machanavajjhala2007diversity,li2007t} do not support participation anonymity, but rather hide sensitive information in the dataset before being used in training. At the same time, existing solutions for anonymous client participation have several limitations~\cite{domingo2021secure,li2021privacy,hasirciouglu2021private,zhao2021anonymous,chen2022fedtor}: they either rely on fixed psuedonyms that are susceptible to traffic analysis attacks~\cite{domingo2021secure}, assume a trusted party to mediate communication~\cite{zhao2021anonymous}, or are vulnerable to man-in-the-middle attacks~\cite{li2021privacy,hasirciouglu2021private}.

\textbf{Dynamic settings.} Allowing clients to (re)join or leave at any time is invaluable for training tasks targeting dynamic environments. A decentralized activity as federated learning may deal with heterogeneous settings involving weak clients who may use low-power devices or have unstable network connectivity. Even it could be the case that clients simply change their minds and abort the training protocol after it starts. The ability to support this dynamicity at low overhead promotes participation, but it is a more challenging setup for client anonymity. 

Most privacy-preserving federated learning solutions do not support dynamic participation; usually clients must join at the onset of the training process during the setup phase. Several solutions support client dropouts (at a relatively high overhead by employing highly interactive recovery protocols that reveal dropout identities)~\cite{bonawitz2017practical,truex2019hybrid,so2021turbo,yang2021lightsecagg} but not addition, or support both but at the expense of a constrained setup that places a cap on the number of clients who can participate~\cite{xu2019hybridalpha}. To the best of our knowledge, supporting anonymity in a dynamic environment has been absent from the current state-of-the-art.

\textbf{An open question.} Therefore, we ask the following question: \emph{can we achieve private federated learning that supports users' anonymity in both static and dynamic settings?}

\subsection{Our Contributions} 
In this paper, we answer this question in the affirmative and propose a system called \sysname that fulfills the requirements above. In particular, we make the following contributions.

{\bf System design.} \sysname utilizes various cryptographic primitives and privacy techniques to achieve its goals. To address anonymity, our system combines a public bulletin board, cryptographic commitments, non-interactive zero-knowledge proofs, and differential privacy such that users can participate in training without revealing their identities. This involves (1) an anonymous registration process guaranteeing that only legitimate clients with honestly-generated datasets can participate, and (2) an unlinkable model updates submission that cannot be traced back to the client. We rely on anonymity sets and zero-knowledge proofs to achieve this, where a client proves owning a legitimate dataset (during registration), and being one of the registered clients (during model updates submission), without revealing anything about their identity or registration information. Moreover, to address membership, inference, and model inversion attacks which could also compromise participation anonymity, we employ differential privacy. A client, after training the model locally and before encrypting the model updates, will sample a noise value and add it to the model updates before encrypting them. The value of this noise is adaptive; it decreases as the number of active clients (those who submitted updates so far in iteration) increases. This reduces the impact on training accuracy without violating the privacy leakage guarantees obtained by differential privacy as elaborated later.

To support dynamic user participation and secure aggregation of model updates, our system employs threshold homomorphic encryption. It splits the roles of a model owner from the aggregators (where we use a committee of aggregators to distribute trust). Aggregators receive encrypted model parameter updates (gradients) from users, and at the end of a training iteration, they operate on these ciphertexts by (homomorphically) adding them to produce a ciphertext of the aggregation. Afterwards, the aggregators decrypt the result and send the aggregated plaintext updates to the model owner so that a new training iteration can be started. Due to this configuration, the clients are not involved in the aggregation or decryption processes and can (re)join and leave training at any time without interrupting the system operation. Furthermore, the bulletin board provides a persistent log accessible to all parties, and facilitates indirect communication between them to reduce interaction.

{\bf Formal security notions and analysis.} We define a notion for private and anonymous federated learning that encompasses three properties: correctness, anonymity, and dataset privacy. Then, we formally prove the security of \sysname based on this notion. To the best of our knowledge, we are the first to provide such formal definition covering anonymity, and the first to build a provably secure client anonymity solution for private federated learning. Our notion could be of independent interest as it provides a rigorous foundation for other anonymity solutions to prove their security guarantees.

{\bf Implementation and evaluation.} To show practicality, we implement \sysname  and empirically evaluate its performance covering different federated learning tasks. We demonstrate scalability of our system by testing scenarios that involve large numbers of clients and contemporary models---the benchmarked architectures are the largest studied in privacy-preserving federated learning literature.
We also show that the augmented components to support privacy and anonymity add reasonable overhead to client runtime.  
For example, in a network of 512 clients, the client setup needed to join the training task takes less than 3 sec, and each training iteration for MNIST classification takes the client a total of 3.2 sec.
We also compare our system to prior work on privacy-preserving federated learning. For our largest benchmark on SqueezeNet architecture trained over TinyImageNet dataset with a network of 512 clients, \sysname is only 1.3$\times$ slower to finish a training round (in 17.5 sec) compared to Truex et al.~\cite{truex2019hybrid} and 2.8$\times$ slower than Bonawitz et al.~\cite{bonawitz2017practical} (the former is a non-interactive private scheme while the latter is interactive, and neither support anonymity). Additionally, we evaluate the accuracy of models trained with \sysname and show that compared to a non-DP baseline, we obtain models with $<0.5\%$ accuracy loss in both independent and identically distributed (IID) datasets between clients and non-IID settings.

\section{A Security Notion for Private and Anonymous Federated Learning}
\label{sec:security-def}
In this section, we define a formal security notion for a private and anonymous federated learning scheme (PAFL). This notion, and its correctness and security properties, are inspired by~\cite{zerocash,quisquis,zether,kairouz2021distributed}. \medskip

\noindent\textbf{Notation.} We use $\lambda$ to denote the security parameter, $\alpha$ to denote correctness parameter, $\gamma$ to denote the privacy advantage of the adversary ($\alpha$ and $\gamma$ are the parameters of the technique used to address membership attacks, if any) $\negl(\cdot)$ to denote negligible functions, and boldface letters to represent vectors. The variable $\stt$ represents the system state, including the data recorded on the bulletin board (these posted by all parties, and the public parameters of the cryptographic building blocks). The notation $\adv^{\mathcal{O}}$ means that an entity, in this case the adversary $\adv$, has an oracle access to $\mathcal{O}$. Lastly, $\xleftarrow{\$}$ denotes drawn at random, and $\ppt$ is a shorthand for probabilistic polynomial time.

\begin{definition}[$(\alpha, \gamma)$-PAFL Scheme]\label{def:pafl} Let $\Pi$ be a protocol between a server $S$, set of aggregators $\agg$, and a set of clients $\clients$ such that each client $\cl_i \in \clients$ holds a dataset $D_i$. Let $\mathbf{M}$ be the initial model that $S$ wants to train, $\mathbf{M}_{actual}$ be the model produced by training $\mathbf{M}$ over $D_i$ (in the clear) for $i = 1, \dots, |\clients|$, and $\mathbf{M}_{\Pi}$ be the trained model produced by the protocol $\Pi$. $\Pi$ is a private and anonymous federated learning (PAFL) scheme, parameterized by bounds $\alpha$ and $\gamma$, if it satisfies the following properties for every $\mathbf{M}$:
\begin{itemize}
\itemsep-0.2em
\vspace{-4pt}
\item \textbf{$\alpha$-Correctness}: The model trained by $\Pi$ achieves an error bound $\alpha$ with high probability compared to the actual model. That is, for $\alpha \geq 0$, and an error function $\error$, with high probability, we have $\error(\mathbf{M}_{actual}, \mathbf{M}_{\Pi}) \leq \alpha$.

\item \textbf{Anonymity}: Any $\ppt$ adversary $\adv$ has a negligible advantage in winning the anonymity game $\anongame$. Formally, for a security parameter $\lambda$, there exists a negligible function $\negl$ such that $\adv$ wins $\anongame$ with probability at most $\frac{1}{2} + \negl(\lambda)$, where the probability is taken over all the randomness used by $\mathcal{A}$ and $\Pi$.

\item \textbf{$\gamma$-Dataset Privacy}: Any $\ppt$ adversary $\adv$ has a negligible additional advantage over $\gamma$ in winning the dataset indistinguishability game $\dindgame$. Formally, for a security parameter $\lambda$ and $\gamma \geq 0$, there exists a negligible function $\negl$ such that $\adv$ wins $\dindgame$ with probability at most $\frac{1}{2} + \gamma + \negl(\lambda)$, where the probability is taken over all randomness used by $\mathcal{A}$ and $\Pi$.

\end{itemize}
\end{definition}

Intuitively, a PAFL scheme is one that is correct and provides anonymity and dataset privacy for clients. Ideally, correctness guarantees that the outcome of a PAFL scheme (i.e., the final trained model) is identical to what will be produced by a training scheme that gets full access to the clients' datasets. Anonymity means that no one can tell whether a client has registered or participated in any training iteration. In other words, a submitted model updates, or any other information a client uses for registration, cannot be traced back to this client. Dataset privacy means that no additional information will be leaked about the private datasets of honest clients beyond any prior knowledge the adversary has.

To make our definition more general, we account for the use of non-cryptographic privacy techniques, such as DP, that may result in accuracy and privacy loss. We do that by parameterizing our definition with $\alpha$ and $\gamma$ standing for correctness (or accuracy loss) and indistinguishability (or privacy loss) parameters, respectively. Having $\alpha = \gamma = 0$ reduces to the ideal case in which $\mathbf{M}_{actual} = \mathbf{M}_{\Pi}$, and an adversary has negligible advantage in breaking anonymity and data set privacy. The bounds for these parameters are derived based on the non-cryptographic privacy technique employed in the system.

We define two security games to capture anonymity and dataset privacy, denoted as $\anongame$ and $\dindgame$, respectively, and the interfaces offered by a PAFL scheme. All parties receive the security parameter $\lambda$, and are given an oracle access to $\pafl$. $\pafl$ maintains the state of the system, including the set of registered clients and aggregators, and any additional information recorded on the board. Since the goal is to protect clients from the model owner $S$, we assume $\adv$ controls $S$ and any subset of clients and aggregators. That is, $\adv$ can register any client or aggregator committee $\agg$ in the system, and can corrupt any of the registered clients and aggregators. $\pafl$ supports the following query types:
\begin{itemize}
\itemsep-0.2em
\vspace{-4pt}
\item $(\setup, 1^{\lambda})$: takes the security parameter as input and sets up the system accordingly---creating the bulletin board, the public parameters needed by all parties/cryptographic building blocks, and the bounds/parameters needed by any additional non-cryptographic privacy technique employed in the system. This command can be invoked only once.

\item $(\register, p, \aux)$: registers party $p$ that could be a client or an aggregators committee. The field $\aux$ specifies the party type and its input: if $p$ is a client, then $\aux$ will include its certified dataset and the certification information, while if $p$ is an aggregator committee, $\aux$ will include the committee's public key. This command can be invoked anytime and as many times as desired. 

\item $(\train, \cl, \aux)$: instructs a (registered) client $\cl$ to train the model using its dataset and submit the model updates. The field $\aux$ defines the dataset that belongs to $\cl$ (the exact information is based on $\Pi$). This command can be invoked anytime and as many times as desired.

\item $(\access)$: returns the updated model (after aggregating all submitted individual model updates received in an iteration). For any iteration, this command can be invoked only once and only at the end of that iteration.

\item $(\corrupt, p, \aux)$: This allows $\adv$ to corrupt party $p$, which could be a client or an aggregator. If $p$ is a client, then $\aux$ will be the registration information of this client (e.g., in \sysname, it is the dataset commitment as we will see later), while if $p$ is an aggregator, $\aux$ will be the public key of that party. This command can be invoked at anytime and as many times as $\adv$ wishes. 
\end{itemize}

Note that the notion $\cl$ only represents the type of a party to be a client, it does not contain its real identity.

Accordingly, the $\anongame$ proceeds as follows:
\begin{enumerate}
\itemsep-0.2em
\vspace{-6pt}
\item $b \xleftarrow{\$} \{0, 1\}$
\item $\stt \leftarrow (\setup, 1^{\lambda})$
\item $(\cl_0, \aux_0, \cl_1, \aux_1) \leftarrow \adv^{\pafl}(1^{\lambda}, \stt)$
\item $(\train, \cl_b, \aux_b)$
\item $\adv$ continues to have access to $\pafl$
\item At the end, $\adv$ outputs $b'$, if $b' = b$ and:
\begin{enumerate}
\itemsep-0.3em
\vspace{-4pt}
    \item both $\cl_0$ and $\cl_1$ are honest,
    \item and there is at least two honest clients participated in every training iteration, and that these clients remained to be honest until the end of the game,
\end{enumerate}
\vspace{-4pt}
then return 1 (meaning that $\adv$ won the $\anongame$), otherwise, return 0.
\end{enumerate}

Note that $\adv$ has access to the current state of the system at anytime, and can see the updated bulletin board after the execution of any command. Also, $\adv$ can see all messages sent in the system, and can access the updated model at the end of any iteration, before and after submitting the challenge. For the chosen clients, $\aux_i$ represents their registration information (which does not include the client identity or its actual dataset $D_i$).\footnote{If the adversary knows the dataset of a client, then it knows that this client is part of the population, i.e., this client suffers from illness, for example; there is no point of hiding whether that client participated in training or not.} Since $\adv$ can choose any two clients for the challenge, $\anongame$ also captures anonymity of registration. That is, if the registration information can be linked to a model update submission, then $\adv$ can always win the game.

$\anongame$ game includes several conditions to rule out trivial attacks. The two clients that $\adv$ selects must be honest, otherwise, if any is corrupt, it will be trivial to tell which client was chosen. Furthermore, since the aggregated model updates is simply the summation of these updates, if only $\cl_b$ participates in the iteration during which the challenge command is executed, it might be trivial for $\adv$ to win (same for any other iteration if only one honest client participates). This is because $\adv$ can access the updated model at the end of that iteration and can extract $\cl_b$'s individual updates. Thus, we add the condition that there must be at least two honest clients have participated in any iteration, and these have to be honest until the end of the game. (Also, any $\agg$ can have at most $n-t$ corrupt parties, which is implicit in $\adv$'s capabilities.)

The $\dindgame$ proceeds as follows:
\begin{enumerate}
\itemsep-0.3em
\vspace{-6pt}
\item $b \xleftarrow{\$} \{0, 1\}$
\item $\stt \leftarrow (\setup, 1^{\lambda})$
\item $(D_0, \aux_0, D_1, \aux_1) \leftarrow \adv^{\pafl}(1^{\lambda}, \stt)$
\item $(\register, \cl, (D_b, \aux))$, $(\train, \cl, \aux)$
\item $\adv$ continues to have access to $\pafl$
\item $\adv$ outputs $b'$, if $b' = b$, and: 
\begin{enumerate}
\itemsep-0.2em
\vspace{-4pt}
    \item $\cl$ is honest,
    \item and there is at least two honest clients participated in every training iteration, and that these clients remained to be honest until the end of the game,
\end{enumerate}
\vspace{-4pt}
then return 1 (meaning that $\adv$ won the $\dindgame$), otherwise return 0.
\end{enumerate}

This game follows the outline of $\anongame$, but with a different construction of the challenge command to reflect dataset privacy. In particular, $\adv$ chooses two valid datasets (the field $\aux$ contains all information required to verify validity). The challenger chooses one of these datasets at random (based on the bit $b$), queries $\pafl$ to register a client using this dataset, and then instructs this client to train the model using the dataset $D_b$. Note that $\aux$ in line 4 is the registration information of the client constructed based on the $\Pi$ scheme. $\adv$ continues to interact with the system, and access the updated model at the end of any iteration. As before, conditions are added to rule out trivial attacks. $\adv$ wins the game if it guesses correctly which of the datasets was chosen in the challenge. Being unable to guess after seeing the outcome of the challenge $\train$ command, and even after accessing the aggregated model updates, means that a PAFL schemes does not reveal anything about the underlying datasets.

\begin{remark} Our PAFL notion can be further generalized to have only the model owner $S$, i.e., no aggregators, so this party is the one who aggregates the individual model updates as well. Moreover, if preserving the privacy of the model is required, then our notion can be extended with a model privacy property to reflect that. Since model privacy is outside the scope of this work, we did not include this property to keep the definition simple.
\end{remark}

\begin{remark}
It should be noted that our definition of anonymity (and so our scheme that satisfies this notion) does not leak negative information (i.e., a client has not participated in training). Both participation and the absence of participation are protected, i.e., identities of those who participate or do not participate are not revealed. 
\end{remark}
\section{Building Blocks}
\label{sec:prelim}
In this section, we provide a brief background on the building blocks employed in \sysname, covering  all cryptographic primitives that we use and the technique of differential privacy.

\vspace{4pt}
\noindent{\bf Commitments.} A cryptographic commitment scheme allows hiding some value that can be opened later. It consists of three $\ppt$ algorithms: $\mathsf{Setup}$, $\mathsf{Commit}$, and $\mathsf{Open}$. On input the security parameter $\lambda$, $\mathsf{Setup}$ generates a set of public parameters $\mathsf{pp}$. To commit to a value $x$, the committer invokes $\mathsf{Commit}$ with inputs $\mathsf{pp}$, $x$, and randomness $r$ to obtain a commitment $c$. $\mathsf{Open}(\mathsf{pp},c)$ opens a commitment by simply revealing $x$ and $r$. Anyone can verify correctness of opening by computing $c' = \mathsf{Commit}(\mathsf{pp}, x, r)$ and check if $c = c'$.

A secure commitment scheme must satisfy: \emph{hiding}, meaning that commitment $c$ does not reveal any information about $x$ beyond any pre-knowledge the adversary has, and \emph{binding}, so a commitment $c$ to $x$ cannot be opened to another value $x' \neq x$ (formal definitions can be found in~\cite{goldreich07}). These security properties enable a party to commit to their inputs (i.e., private datasets in our case), and publish the commitment publicly without exposing the private data.

\vspace{4pt}
\noindent{\bf Threshold homomorphic encryption.} Homomorphic encryption allows computing over encrypted inputs and producing encrypted outputs. Such operations include homomorphic addition and multiplication. That is, let $ct_1$ be a ciphertext of $x_1$, and $ct_2$ be a ciphertext of $x_2$, then $ct_1 + ct_2$ produces a ciphertext of $x_1 + x_2$, and $ct_1 \cdot ct_2$ produces a ciphertext of $x_1 \cdot x_2$ (the exact implementation of the homomorphic '+' and '$\cdot$' vary based on the encryption scheme). Some encryption schemes support only one of these operations, e.g., Paillier scheme~\cite{paillier1999public} is only additively homomorphic. Supporting both addition and multiplication leads to fully homomorphic encryption~\cite{bgv,mukherjeewichs}. Since we focus on secure aggregation of model updates, we require only additive homomorphism.

A homomorphic encryption scheme is composed of three $\ppt$ algorithms: $\mathsf{KeyGen}$ that generates encryption/decryption keys (and any other public parameters), $\mathsf{Encrypt}$ which encrypts an input $x$ to produce a ciphertext $ct$, and $\mathsf{Decrypt}$ which decrypts a ciphertext $ct$ to get the plaintext input $x$. Correctness states that $\mathsf{Decrypt}$ produces the original plaintext for any valid ciphertext produced by $\mathsf{Encrypt}$, and in case of homomorphic operations, the correct result (add and/or multiply) is produced. Security is based on the regular security notion for encryption (i.e., semantic security or indistinguishability-based). In this work, we require indistinguishability against CPA (chosen-plaintext attacker).

The threshold capability is related to who can decrypt the ciphertext. To distribute trust, instead of having the decryption key known to a single party, it is shared among $n$ parties. Thus, each of these parties can produce a partially decrypted ciphertext upon calling $\mathsf{Decrypt}$, and constructing the plaintext requires at least $t$ parties to decrypt. In a threshold homomorphic encryption scheme, $\mathsf{KeyGen}$ will be a distributed protocol run by the $n$ parties to produce one public key and $n$ shares of the secret key, such that each party will obtain only her share (and will not see any of the others' shares). In \sysname, we use the threshold Paillier encryption scheme~\cite{damgaard2001generalisation}.

\vspace{4pt}
\noindent{\bf Zero-knowledge proofs.} A (non-interactive) zero-knowledge proof (ZKP) system allows a prover, who owns a private witness $\omega$ for a statement $x$ in language $\mathcal{L}$, to convince a verifier that $x$ is true without revealing anything about $\omega$. A ZKP is composed of three $\ppt$ algorithms: $\mathsf{Setup}$, $\mathsf{Prove}$, and $\mathsf{Verify}$. On input a security parameter $\lambda$ and a description of $\mathcal{L}$, $\mathsf{Setup}$ generates public parameters $\mathsf{pp}$. To prove that $x \in \mathcal{L}$, the prover invokes $\mathsf{Prove}$ over $\mathsf{pp}$, $x$, and a witness $\omega$ for $x$ to obtain a proof $\pi$. To verify this proof, the verifier invokes $\mathsf{Verify}$ over $\mathsf{pp}$, $x$, and $\pi$, and accepts only if $\mathsf{Verify}$ outputs 1. In general, all conditions needed to satisfy the NP relation of $\mathcal{L}$ are represented as a circuit. A valid proof will be generated upon providing valid inputs that satisfy this circuit. Some of these inputs could be public, which the verifier will use in the verification process, while others are private, which constitute the witness $\omega$ that only the prover knows.

A secure ZKP system must satisfy completeness, soundness, and zero-knowledge. Completeness states that any proof generated in an honest way will be accepted by the verifier. Soundness ensures that if a verifier accepts a proof for a statement $x$ then the prover knows a witness $\omega$ for $x$. In other words, the prover cannot convince the verifier with false statements. Finally, zero-knowledge ensures that a proof $\pi$ does not reveal anything about the witness $\omega$. Many ZKP systems added a succinctness property, so that proof size is constant and verification time is linear the input size and independent the NP relation circuit size. These are called zero-knowledge succinct non-interactive argument of knowledge (zk-SNARKs). Formal definitions of ZKP systems can be found in~\cite{goldreich07,bitansky13}. In \sysname, we use the proof system proposed in~\cite{groth2016size}.

\vspace{4pt}
\noindent{\bf Differential privacy.} Differential privacy (DP) is a technique usually used to address attacks such as membership and inference attacks. That is, knowing a data point and the trained model, an attacker can tell if this datapoint was used in training the model. DP provides the guarantee that inclusion of a single instance in the training datasets causes a statistically insignificant change to the training algorithm output, i.e., the trained model. Thus, it limits the success probability of the attacker in membership attacks. Formally, DP is defined as follows~\cite{dwork2014algorithmic} (where $\epsilon > 0$ and $0 < \delta < 1$):

\begin{definition}[Differential Privacy] A randomized mechanism $\mathcal{K}$ provides $(\epsilon, \delta)$-differential privacy, if for nay two datasets $D_0$ and $D_1$ that differ in only single entity, for all $R \subseteq Range(\mathcal{K})$: $\Pr[\mathcal{K}(D_0) \in R] \leq e^{\epsilon}\Pr[\mathcal{K}(D_1) \in R] + \delta$
\end{definition}

We adopt the Guassian DP mechanism as described in~\cite{wei2020federated}, and apply the optimization for the secure aggregation setting in~\cite{truex2019hybrid,kairouz2021distributed}. The clients sample a noise from a Gaussian distribution $\mathcal{N}(0, \sigma^2)$ and add the noise to their model updates before encrypting and submitting them. 
Each client $i$ clips their model gradients $\bm{g}_i$ to ensure $\parallel \bm{g}_i\parallel < C$ where $C$ is a clipping threshold for bounding the norm of gradients. Client $i$ then sets sensitivity $S_f = 2C/|D_i|$ where $D_i$ is the $i-$th client's dataset, assuming gradients are shared after one epoch of local training. The noise scale is set as $\sigma \geq cTS_f/\epsilon$ where $c \geq \sqrt{2\ln{(1.25/\delta)}}$ and $T$ indicates exposures of local parameters (number of epochs or iterations), and satisfies $(\epsilon, \delta)$-DP~\cite{dwork2014algorithmic,abadi2016deep, wei2020federated} for $\epsilon < 1$.
We apply the optimization in~\cite{truex2019hybrid} by dividing the noise scale by the number of participating clients. Since only the aggregation of the updates is revealed, the individual noise value added by a client can be reduced while guaranteeing that the aggregate value maintains the desired level of privacy, leading to better accuracy.

Using these parameters, the advantage $\gamma$ of the attacker in membership attacks can be computed as~\cite{humphries2020investigating}:
\begin{equation}\label{eq:gamma}
    \gamma \leq (1 - e^{-\epsilon} + 2\delta)(e^{\epsilon} + 1)^{-1} 
\end{equation}

As for the error bound $\alpha$, the Gaussian mechanism satisfies an absolute error bound $\alpha \leq O(R\sqrt{\log k})$, where $k$ is the number of queries and $R := S_f \sqrt{k \log 1/\delta}/\epsilon$~\cite{dwork2006our,dagan2022bounded} (a query here indicates one training round, so $k$ is the number of training iterations, and $S_f$ is the query sensitivity, which is the sensitivity of the training function). We use these bounds in our security proofs in Appendix~\ref{app:proof}. 

In describing our design, we use DP as a blackbox. A client invokes function called $\mathsf{DP.noise}$ to sample the appropriate noise value. This makes our design modular as any secure DP-mechanism, other than the one we employ, can be used.
\section{\sysname Design}
\label{sec:design}
\sysname relies on four core ideas to achieve its goals: certification of clients' datasets to prevent impersonation, anonymity sets to support anonymous registration and model training, a designated aggregators committee to prevent accessing the individual model updates submitted by the clients, and a public bulletin board to facilitate indirect communication and logging. In this section, we present the design of \sysname showing the concrete techniques behind each of these ideas and how they interact with each other.

\subsection{System Model}
\label{sec:sys-model}
As shown in Figure~\ref{fig:anofel-model}, for any learning task there is a model owner $S$, a model aggregator set (or committee) $\agg$ of size $n$ (we use a committee instead of a single aggregator to distribute trust), and a set of clients $\clients$ who wants to participate in this learning task. During a training iteration, $\clients$ will train the initial model locally over their private datasets and publish encrypted updates. $S$ has to wait until the iteration is finished, after which $\agg$ will aggregate the updates and grant $S$ access to the aggregated value. \sysname can support any aggregation method based on summation or averaging (e.g., FedSGD~\cite{chen2016revisiting} and FedAVG~\cite{mcmahan2017communication}) given that the number of participants is public. All parties have access to a public bulletin board, allowing them to post and retrieve information about the training process. There could be several learning tasks going on in the system, each with its own $S$, $\agg$, and $\clients$, and all are using the same bulletin board.\footnote{In fact, it could be the case that the same parties are involved in several learning tasks, but they track each of these separately.} In the rest of this section, we focus on one learning task to explain our protocol; several learning tasks will separately run the same protocol between the involved parties.

\begin{figure}[t!]
\centerline{
\includegraphics[height= 3.0in, width = 0.9\columnwidth]{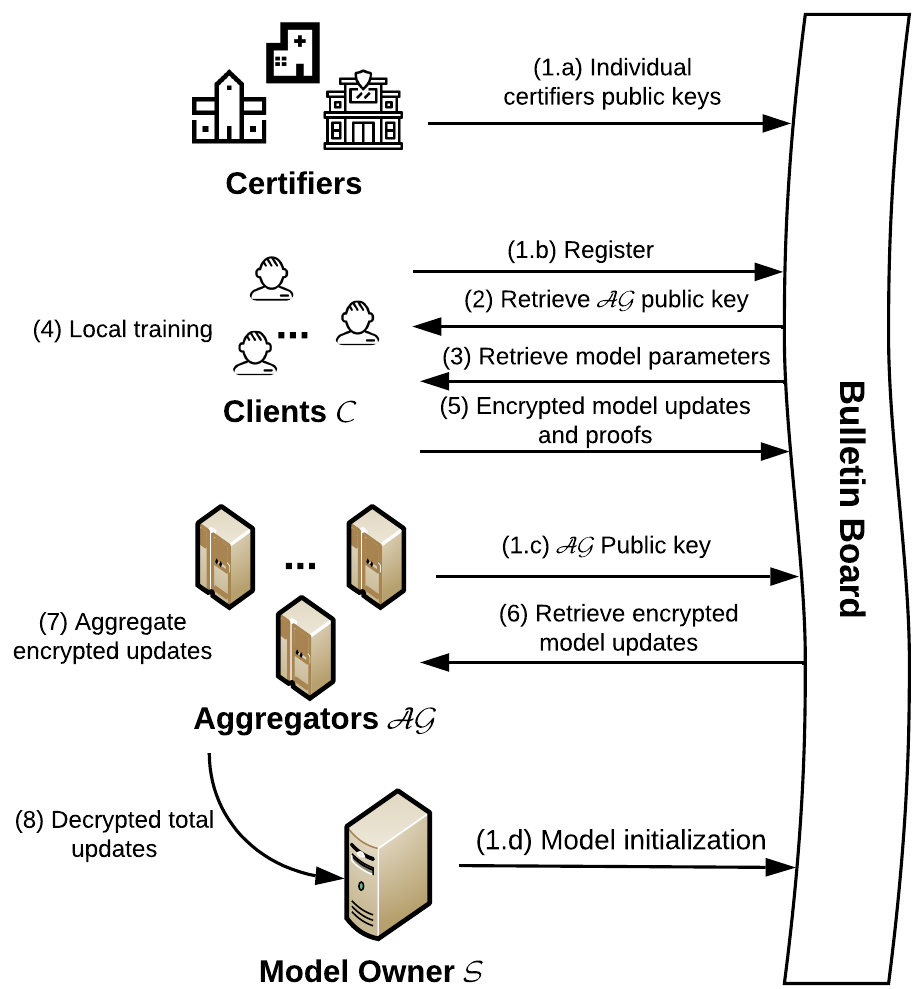}}
\caption{\sysname system architecture. All these steps (except 1.a - 1.c, and 2) will be repeated for each training iteration.} 
\label{fig:anofel-model}
\end{figure}

\subsection{Threat Model} 
We assume a secure and immutable public bulletin board available to all parties, which accepts only authenticated information that complies with predefined correctness rules.\footnote{This can be instantiated in a decentralized way using a blockchain with miners, or validators, verifying correctness.} 

We adopt the following adversary model (we deal with $\ppt$ adversaries). For clients, we assume them to be malicious during registration (a malicious party may behave arbitrarily), while we assume these clients to be semi-honest during training (a semi-honest party follows the protocol but may try to collect any additional information). Thus, during registration, a client with an invalid (or poisoned) dataset may try to register, while during training, registered clients will use their valid (registered) datasets in training and submit valid updates. We assume the server $S$ to be always malicious, so it may try to impersonate clients in the registration phase, collude with aggregators during the training phase, or manipulate the model posted at the beginning of each iteration. For $\agg$, we assume that at maximum $n - t$ parties can be malicious, where $t$ is the threshold required for valid decryption. 

Since our goal is supporting anonymity, we do not assume any authenticated communication channels. Lastly, we work in the random oracle model where hash functions are modeled as random oracles.

\subsection{System Workflow}
\label{sec:sys-workflow}
\sysname achieves anonymity and privacy by combining a set of cryptographic primitives, such as threshold homomorphic encryption and zero-knowledge proofs, differential privacy, and a public bulletin board. The latter is used to facilitate indirect communication between the parties and to create anonymity sets to disguise the participants. Our techniques of combining zero-knowledge proofs and anonymity sets are inspired by recent advances in private and anonymous cryptocurrencies~\cite{zerocash,quisquis,zether}. 

Each client must register before participation (step 1.b in Figure~\ref{fig:anofel-model}) by publishing on the board a commitment to the master public key and dataset this client owns (the commitments are never revealed and don't leak any information). Similarly, the aggregators $\agg$ must register by posting their public key on the board (step 1.c in Figure~\ref{fig:anofel-model}), which will be used by the clients to encrypt their model updates. As we use DP to protect against membership and inference attacks, the client samples a noise and adds it to their  updates before encrypting them. The encrypted updates will be accompanied with a zero-knowledge proof (ZKP) attesting that a client is a legitimate and registered data owner, but without revealing the public key or the dataset commitment of this client. Thus, anonymity is preserved against everyone (the server, aggregators, other clients, or any other party).

As shown in the figure (steps 1.d and 3), a server publishes the initial model updates on the bulletin board, which are retrieved by the clients to be used in the local training. The use of this immutable public bulletin board avoids direct communication between the server and clients, which would compromise anonymity. Moreover, it addresses privacy attacks resulting from distributing different initial model parameters to clients (this type of attacks has been recently demonstrated in~\cite{pasquini2021eluding}). In fact, \sysname is the only private federated learning system that enjoys this advantage.

To prevent Sybil and data poisoning attacks, each dataset must be certified by its source, e.g., a hospital. Since exposing the certifier reveals the dataset type and impacts privacy, \sysname creates an anonymity set for the certifiers by having all their public keys posted on the board (step 1.a in Figure~\ref{fig:anofel-model}). Thus, during the registration, a client will provide a ZKP that its hidden (committed) dataset is signed by one of the certifiers---without specifying which one.

\begin{figure*}[ht!]
\centerline{
\includegraphics[height= 2.2in, width = 0.9\textwidth]{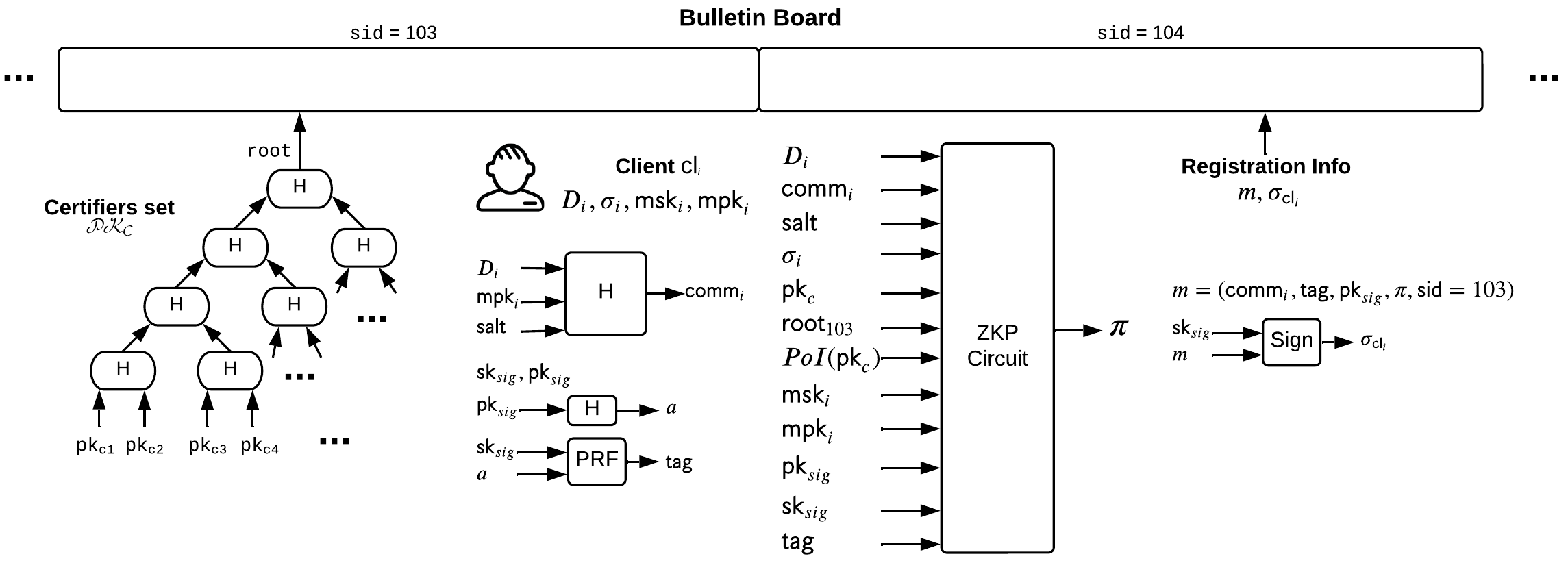}}
\vspace{-8pt}
\caption{Client setup---anonymous registration process.} 
\label{fig:client-reg}
\end{figure*}

Furthermore, to hide which training activity a client is participating in, which if revealed it will expose the dataset type, \sysname creates an anonymity set for the aggregators. That is, the system will have several ongoing training activities, each of which with its own $\agg$ committee. When submitting a model update, the client will choose a subset of these aggregators including the target $\agg$ who is managing the training activity the client is interested in. The client then encrypts the updates under the public keys of this subset---encrypt the actual updates for the target $\agg$ while encrypt 0 for the rest of the aggregators. Consequently, even if it is revealed that a client has submitted a model update, this will not expose which training activity this client is part of.

Accordingly, \sysname proceeds in three phases: setup, model training, and model access, which we discuss below.

\subsubsection{Setup}
\label{subsec:setup}
The initial system setup includes creating the bulletin board, generating all public parameters needed by the cryptographic primitives used, and configuring the parameters for noise distribution and privacy/accuracy loss bounds of DP that all parties will use. Then, the certifiers, aggregators, and each client must run the setup process.\footnote{A PKI is needed to ensure the real identities of the certifiers, aggregators, and model owner. So, when any of these entities posts its public key on the board, this must be accompanied with a certificate (from a certificate authority) to prove that indeed the party owns this key.}

\vspace{4pt}
\noindent\textbf{Certifiers.} Each certifier posts its public key on the board. Let $\mathcal{PK}_C$ be the set of all certifiers' public keys. 

\vspace{4pt}
\noindent\textbf{Aggregators.} $\agg$ run the setup of a threshold homomorphic encryption to generate a public key and shares of the secret key. They post the public key on the board, while each party keeps its secret key share. To authenticate the public key, at least $t$ of these aggregators must post it on the board (each posting is signed using the aggregator's own public key).\footnote{The individual public keys of $\agg$ are managed by a PKI.} 

\vspace{4pt}
\noindent\textbf{Clients.} The setup process for clients is more involved compared to the previous entities. This is a natural result of supporting anonymity. As shown in Figure~\ref{fig:client-reg}, which is the detailed version of step 1.b in Figure~\ref{fig:anofel-model}, each client $\cl_i \in \clients$, with a dataset $D_i$ and a master keypair $(\msk_i, \mpk_i)$,\footnote{Clients need a PKI for their master public keys, so a certifier can check that a client owns the presented master public key. However, to preserve anonymity, these are hidden in the training process, and a certifier cannot link an encrypted model update to the master public keys.} obtains a certificate $\sigma_i$ from its certifier---$\sigma_i$ could be simply the certifier's signature over $D_i \concat \mpk_i$. Then, $\cl_i$ commits to its dataset $D_i$ (without revealing anything about it) as follows: 
\begin{itemize}
\itemsep-0.3em
\vspace{-4pt}
\item Compute a commitment to $D_i$ and $\mpk_i$ as $\comm_i = H(D_i \concat \mpk_i \concat \salt)$, where $H$ is a collision resistant hash function and $\salt$ is a fresh random string in $\zo^{\lambda}$. 

\item Generate a fresh digital signature keypair $(\pk_{sig}, \sk_{sig})$. Compute $a = H(\pk_{sig})$ and $\taag = \prf_{\msk_i}(a)$, where $\prf$ is a pseudorandom function, and $\taag$ serves as an authentication tag over the fresh key to bind it to the master keypair of the client. Similar to~\cite{zerocash}, we instantiate the PRF as $\prf_{\msk_i}(a) = H(\msk_i \concat a)$. 

\item Generate a ZKP $\pi$ to prove that the dataset is legit and owned by $\cl_i$. In particular, this ZKP attests to the following statement (again without revealing anything about any of the private data that the client knows): given a commitment $\comm_i$, a signature verification key $\pk_{sig}$, a tag $\taag$, and a bulletin board state index $\sid$, client $\cl_i$ knows a dataset $D_i$, randomness $\salt$, master keypair $(\mpk_i, \msk_i)$, a certifier's key $\pk_c$, a certificate $\sigma_i$, and a signing key $\sk_{sig}$, such that: 
\begin{enumerate}
\itemsep-0.3em
\vspace{-4pt}
    \item $D_i$, $\mpk_i$, and $\salt$ are a valid opening for the commitment $\comm_i$, i.e., $\comm_i = H(D_i \concat \mpk_i \concat \salt)$.\footnote{Note that the opening of the commitment is a private input the client uses to generate the proof (and same applies to the training phase as we will see shortly). This opening is never revealed to the public---The board only has the commitment, and a ZKP on its well-formedness that does not leak any information about any private data used locally to generate the proof.}
    \item $\sigma_i$ has been generated using $\pk_c$ over $D_i \concat \mpk_i$, and that $\pk_c \in \PK_C$ with respect to the set $\PK_C$ registered on the board at state with index $\sid$.
    \item The client owns $\mpk_i$, i.e., knows $\msk_i$ that corresponds to $\mpk_i$.
    \item The tag over the fresh $\pk_{sig}$ is valid, i.e., compute $a = H(\pk_{sig})$ and check that $\taag = \prf_{\msk_i}(a)$.
\end{enumerate}

\item Sign the proof and the commitment: set $m =(\comm_i,\taag,$ $\pk_{sig}, \pi, \sid)$, use $\sk_{sig}$ to sign $m$ and obtain a signature $\sigma_{\cl_i}$.

\item Post $(m, \sigma_{\cl_i})$ on the bulletin board.
\end{itemize}

\begin{figure*}[ht!]
\centerline{
\includegraphics[height= 2.15in, width = 0.93\textwidth]{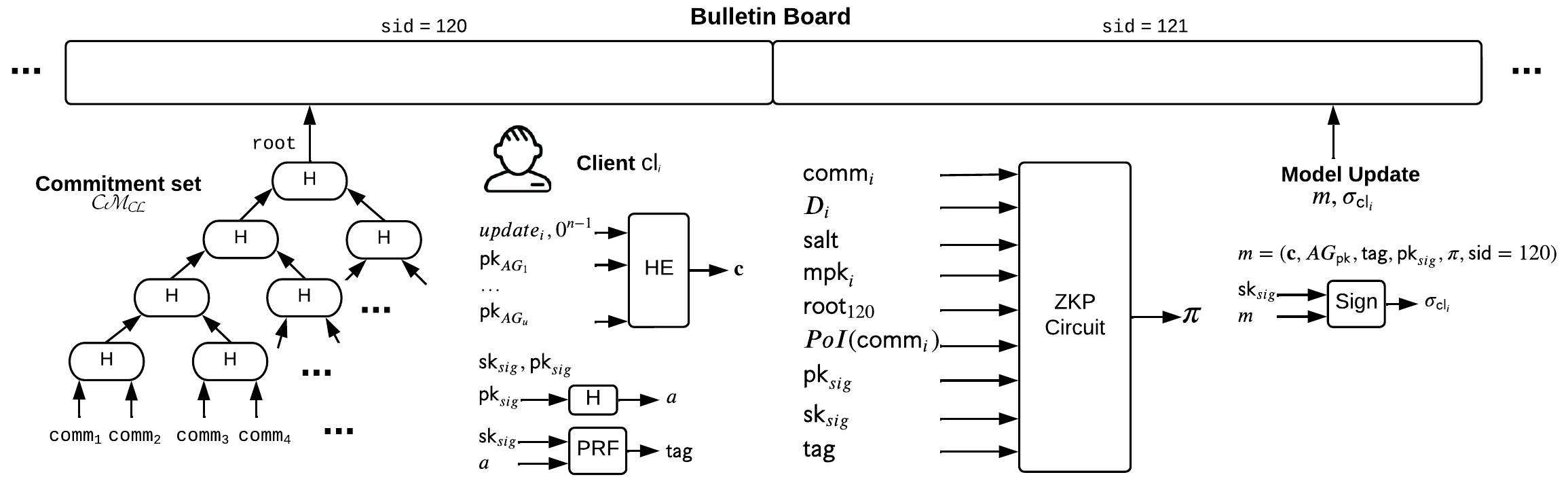}}
\vspace{-4pt}
\caption{Model training---anonymous model updates submission ($0^{n-1}$ means $n-1$ zeros, and HE is homomorphic encryption).} 
\label{fig:training}
\end{figure*}

As for verifying the third condition, i.e., the client knows $\msk_i$, it is simply done by computing the public key based on the input $\msk_i$ and checking if it is equal to $\mpk_i$. Although $\mpk_i$ is not recorded explicitly on the bulletin board, it is bound to the client since it is part of the dataset commitment $\comm_i$ certified by $\sigma_i$.

To allow efficient proof generation with respect to the anonymity set $\PK_C$, a Merkle tree is used to aggregate the key set $\PK_C$ as shown in Figure~\ref{fig:client-reg}. Proving that a key $\pk_c \in \PK_C$ is done by showing a proof of inclusion (PoI) of that key in the tree. In other words, the circuit underlying the ZKP generation takes a membership path of the key and the tree root and verifies the correctness of that path. Thus, the cost will be logarithmic in the set size. The tree can be computed by the entities maintaining the board, with the root published on the board to allow anyone to use it when verifying the ZKP.

Note that a ZKP is generated with respect to a specific state of the anonymity set $\PK_C$. This state is the root of the Merkle tree of this set, which changes when a new certifier joins the system. Such change will invalidate all pending ZKPs, and thus, invalidate all pending client registrations tied to the older state. To mitigate this, a client should specify the state index $\sid$ based on which the ZKP (and hence, the Merkle tree) was generated. So if the board is a series of blocks as in Figure~\ref{fig:client-reg}, $\sid$ is the block index containing the root used in the proof.

All these conditions are modeled as an arithmetic circuit. The client has to present valid inputs that satisfy this circuit in order to generate a valid ZKP. Only registration with valid ZKPs will be accepted, where we let $\CM_{CL}$ be the set of all valid clients' commitments. Registration integrity is preserved by the security of the digital signature $\sigma_{\cl_i}$: if an adversary tampers with any of the information that a client submits---$\comm_i$, $\taag$, $\pk_{sig}$, $\pi$, or $\sid$---this will invalidate $\sigma_{\cl}$ and will lead to rejecting the registration.

Note that clients can perform the setup at anytime, and once their registration information is posted on the board, they can participate in the model training immediately. Thus, \sysname allows clients to join at anytime during the training process, and each client can perform the setup phase on its own.

\subsubsection{Model Training} 
\label{subsec:training}
At the beginning of each training iteration, $S$ posts the initial model parameters on the bulletin board (step 1.d in Figure~\ref{fig:anofel-model}). Each client $\cl_i$ retrieves them and trains the model locally over its dataset (step 3 and 4 in Figure~\ref{fig:anofel-model}, respectively). Then $\cl_i$ shares the model updates privately and anonymously without revealing (1) which training activity it is participating in nor (2) the dataset commitment it owns (i.e., without revealing its identity). Client $\cl_i$ does that as follows (see Figure~\ref{fig:training}, which is the detailed version of step 5 in Figure~\ref{fig:anofel-model}):
\begin{itemize}
\itemsep-0.3em
\vspace{-4pt}
\item Choose a set of aggregators $AG = \{\agg_1, \dots, \agg_u\}$ including the target $\agg$, with $AG_{\pk}$ denoting the public keys of these aggregators. Shuffle $AG$ to avoid any ordering attacks (e.g., if the target $\agg$ is always placed first, this reveals the target training activity).

\item Invoke $\textsf{DP.noise}$ to sample a noise value to be added to the model updates. As mentioned earlier, we apply the optimization in~\cite{truex2019hybrid} by dividing the noise scale by the number of participating clients since the model updates will be decrypted after being aggregated. 

\item Encrypt the model updates under the target $\agg$ public key, while encrypt 0 under the public keys of the rest. This will produce $\mathbf{c}$; a vector of $u$ ciphertexts.\footnote{A client will have a fixed $AG$ selected at the beginning. Changing $AG$ between iterations must be done carefully; for a new $AG'$, if $AG \cap AG' = \agg$, it would be trivial to tell which training activity a client is part of.}

\item Generate a fresh digital signature keypair $(\pk_{sig}, \sk_{sig})$. Compute $a = H(\pk_{sig})$ and $\taag = \prf_{\msk_i}(a)$.

\item Produce a ZKP $\pi$ (with respect to the current state of the board at index $\sid$) attesting that: $\cl_i$ is a legitimate owner of a dataset, and that the fresh digital signature key was generated correctly. Thus, this ZKP proves the following statement: given a signature key $\pk_{sig}$, and a tag $\taag$, $\cl_i$ knows the opening of some commitment $\comm \in \CM_{CL}$ (this proves legitimacy), and that $\taag$ was computed correctly over $\pk_{sig}$ as before. Since we adopt the semi-honest adversary model during the training phase, there is no need to prove that $\mathbf{c}$ has only one non-zero update.\footnote{Note that although we assume semi-honest clients during training, we still need the ZKP above to preserve integrity and make sure only registered clients participate. That is, a malicious adversary (who could be the server) may impersonate a client during training (without doing any registration), and it may alter the submitted updates (thus we need to prove that the signature is honestly generated by a registered client).}

\item To preserve integrity, sign the proof, the ciphertext, and the auxiliary information. That is, set $m = (\mathbf{c}, AG_{\pk}, \taag, \pk_{sig}, \sid, \pi)$ and sign $m$ using $\sk_{sig}$ to a produce a signature $\sigma_{\cl_i}$.

\item Post $(m, \sigma_{\cl_i})$ on the bulletin board.
\end{itemize}

We use the Merkle tree technique to aggregate the commitment anonymity set $\CM_{CL}$. A client provides a proof of inclusion of its commitment in the Merkle tree computed over $\CM_{CL}$ with respect to a specific state indexed by $\sid$. The latter is needed since we allow clients to join anytime, and thus, $\CM_{CL}$, and its Merkle tree, will change over time.

Based on the above, \sysname naturally supports dynamic client participation. As mentioned before, a client who wants to join can do that immediately after finishing the setup. While (registered) clients who do not wish to participate in a training iteration simply do not send any updates. \sysname does not need a recovery protocol to handle additions/dropouts since the setup of a client does not impact the setup of the system, $\agg$, or other clients. Also, the model updates submitted by any client do not impact the updates submitted by others. Not to mention that any information needed to perform the setup is already on the bulletin board, so interaction between the parties is needed. Furthermore, submitting model updates is done in one shot; a client posts $(m, \sigma_{\cl_i})$ on the board. Since we use non-interactive ZKPs, $\agg$, and any other party, can verify the proof on their own.

\subsubsection{Model Access}
\label{subsec:access}
At the end of each training iteration, $\agg$ members retrieve all client updates---those that are encrypted under her public key---from the board, and aggregate them using the additive homomorphism property of the encryption scheme (steps 6 and 7 in Figure~\ref{fig:anofel-model}, respectively). Then, each member decrypts the  ciphertext using its secret key share, producing a partial decryption that is sent to $S$ (step 8 in Figure~\ref{fig:anofel-model}). Once $S$ receives at least $t$ responses, it will be able to construct the plaintext of the aggregated model updates and start a new iteration.

Signaling the end of a training iteration relies on the bulletin board. Adding a future block with a specific index will signal the end. Thus, the system setup will determine the block index of when training starts, and the duration of each iteration (in terms of number of blocks). Since all parties have access to the board, they will be able to know when each iteration is over. Another approach is to simply have $S$ post a message on the board to signal the end of each iteration.

Although \sysname is a system for federated learning that involves several parties, it is not considered an interactive protocol. These parties do not communicate directly with each other---the bulletin board mediates this communication. When sending any message, the sender will post it on the bulletin board, and the intended recipient(s) will retrieve the message content from the board.

\begin{remark}
The concrete instantiation of the bulletin board impacts runtime. The board mediates communication between parties and must verify the validity of all information postings before accepting them. To speed up this process, the board be formed from a sequence of blocks of information maintained by a committee of registered validators (similar to a blockchain but without the full overhead a regular blockchain introduces). Alternatively, a regular blockchain that takes advantage of recent optimized designs and scalability techniques
~\cite{wang2019sok,gilad2017algorand,bagaria2019prism} can be adopted. We note that the concrete instantiation of the board is outside the scope of this work.
\end{remark}

\subsection{Extensions}
\label{subsec:extensions}
\noindent{\bf Addressing a stronger adversary model.}
\sysname assumes semi-honest clients in the training phase. Therefore, mitigating threats of (1) using a legitimate (registered and certified) dataset in a training activity of totally different type---e.g., use medical data to train a model concerned with vehicles, and (2) encrypt non-zero values to other aggregators (other than the target $\agg$) in $AG$ to corrupt their training activities, are out of scope. Nevertheless, we can make our adversary model stronger by considering a semi-malicious client who may attempt these attacks. 

To mitigate the first attack, we require the certifier to add a dataset type $\dt$ to the dataset certificate, and hence, we require the client to prove that it has used a dataset with the correct type in training. That is, each training activity in the system will have a designated type $\dt$, and a certifier will check that a dataset $D_i$ is indeed of type $\dt$ (so it can be used to train any model of type $\dt$) before signing $D_i \concat \mpk_i \concat \dt$. Also, the ZKP circuit a client uses in training must check that the target $\agg$ (that will receive a ciphertext of non-zero value) is managing a training activity with an identical $\dt$. Otherwise, a valid proof cannot be generated. To mitigate the second attack, we can add another condition to be satisfied in the ZKP circuit (again the one a client uses during training) to validate the ciphertext $\mathbf{c}$. That is, verify that only one ciphertext in $\mathbf{c}$ is for a non-zero value while the rest are zeros. This requires providing the ZKP circuit with all plaintexts and the randomness used in encryption, so it can the ciphertexts based on these, denoted as $\mathbf{c}'$, and check if $\mathbf{c}' = \mathbf{c}$.

Addressing malicious clients during training, i.e., these who may deviate arbitrarily from the protocol, can be done in a generic way using ZKPs as in~\cite{burkhalter2021rofl}. That is, a model update will be accompanied with ZKP proofs on well-formedness, meaning that the registered dataset and the actual initial model parameters were used and training was done correctly. Extending \sysname to support that while preserving its efficiency level is part of our future work.

\vspace{4pt}
\noindent{\bf Reducing storage costs of the bulletin board.} ML models may involve thousands of parameters. The server needs to post the initial values of these parameters on the board for each training iteration. Also, a client posts the updated version of all these parameters on the board. This is a large storage cost that may create a scalability problem, e.g., if the board is a blockchain this cost could be infeasible. To address this issue, we can employ any of the solutions currently used by the blockchain community, e.g., store the (ciphertext of) model parameters on a decentralized storage network, and post only the hash of them on the board (with a pointer to where the actual data is stored). Furthermore, once the data is used, i.e., a training iteration concluded, initial model parameters and all clients updates can be discarded, which reduces the storage cost significantly. Note that the initial model parameters are posted by the server, who is not anonymous, while the encrypted updates are posted by clients. Thus, an anonymous off-chain storage must be used (like an anonymous sidechain) to avoid compromising their anonymity.

\subsection{Security of \sysname}
\label{sec:security}
\sysname realizes a correct and secure PAFL scheme based on the notion defined in Section~\ref{sec:security-def}. In Appendix~\ref{app:proof}, we formally prove the following theorem:

\begin{theorem}\label{theorem:pafl-sec}
The construction of \sysname as described in Section~\ref{sec:design} is a correct and secure PAFL scheme (cf. Definition~\ref{def:pafl}).
\end{theorem}
\section{Performance Evaluation}
\label{sec:eval}
In this section, we provide details on the implementation and performance evaluation of \sysname, and benchmarks to measure its overhead compared to prior work.

\subsection{Implementation}
\vspace{-4pt}
For hash functions, we use the Pedersen hash function~\cite{hopwood2020zcash}, 
with an alternative implementation using
Baby-Jubjub elliptic curve~\cite{belles2021twisted} and 4-bit windows~\cite{jubjubImpl}, which requires less constraints per bit than the original implementation.
For threshold additive-homomorphic encryption, we use the threshold version of Paillier encryption scheme~\cite{damgaard2001generalisation} based on~\cite{paillier-impl}. For digital signatures, we use EdDSA~\cite{bernstein2012high} over Baby-Jubjub elliptic curve based on~\cite{ecdsa-impl}.
For zero-knowledge proofs, we use Groth16 zk-SNARKS~\cite{groth2016size} implemented in libsnark~\cite{libsnark}.
We use PyTorch~\cite{paszke2019pytorch} to incorporate differential privacy, and implement the FedSGD~\cite{chen2016revisiting} algorithm for federated learning.

{\textbf{Dataset and Models.}} We evaluate the performance of \sysname using three federated learning tasks. Our first benchmark is LeNet5~\cite{lecun1998gradient}  architecture with 61.7K parameters trained on the MNIST~\cite{lecun1998mnist} dataset.
Our second benchmark is ResNet20~\cite{he2016deep} architecture with 273K parameters trained over the CIFAR10~\cite{krizhevsky2010cifar} dataset. 
Our third benchmark is SqueezeNet~\cite{han2015learning} with 832K parameters trained over TinyImageNet~\cite{yao2015tiny} dataset. This benchmark is the largest studied in private federated learning literature~\cite{chowdhury2021eiffel,truex2019hybrid, burkhalter2021rofl}.
Since batch normalization is not compatible with DP~\cite{yu2021not}, we replace all batch normalization layers with group normalization~\cite{wu2018group} in ResNet20 and SqueezeNet with negligible effect on accuracy.

{\textbf{Configuration.}}
For our runtime experiments, we consider a network of $N=\{16,32,64,128,256,512\}$ clients. We also assume one committee consisting of 3 aggregators. Runtimes are benchmarked on an Intel i9-10900X CPU running at 3.70GHz with 64GB of memory assuming 8 threads, and the mean of 5 runs are reported for each experiment. We present micro-benchmarks of \sysname components as well as end-to-end benchmarks for a training iteration. We also evaluate \sysname accuracy under IID and non-IID dataset settings. The DP parameters are set as $\epsilon=0.9$, $\delta=10^{-5}$, and norm clipping threshold $C=1$ for MNIST and $C=2$ for CIFAR10 and TinyImageNet benchmarks.%

\begin{table}[t!]
\caption{zk-SNARKs runtime for client setup and training.}
\label{tab:zkp-performance}
\vspace{-6pt}
\centering
\resizebox{.4\textwidth}{!}{
\begin{tabular}{lcccc}\toprule
           & \multicolumn{2}{c}{Setup} & \multicolumn{2}{c}{Training} \\
\cmidrule(lr{0.5em}){2-3}
\cmidrule(lr{0.5em}){4-5}
\# Clients & Prove (s)  & Verify  (ms) & Prove (s)    & Verify (ms)   \\
\midrule
16         & 1.97       & 4.65         & 0.79         & 4.67          \\
32         & 2.04       & 4.74         & 0.85         & 4.55          \\
64         & 2.09       & 4.61         & 0.96         & 4.57          \\
128        & 2.16       & 4.62         & 1.04         & 4.53          \\
256        & 2.22       & 4.64         & 1.14         & 4.54          \\
512        & 2.28       & 4.65         & 1.25         & 4.56         \\
\bottomrule
\end{tabular}}
\end{table}

\subsection{Results}
\noindent\textbf{Runtime Overhead.}
Table~\ref{tab:zkp-performance} shows the performance of zero-knowledge proof implementation of clients setup and training circuits. The prove and verify runtimes are reported for different numbers of clients participating in the learning task. The size of proof is a constant 1019 bits. The setup prove overhead is a one-time cost for clients, and the training prove overhead is incurred for each training iteration. As the results suggest, the prover runtime for both setup and training increase sub-linearly with the number of clients, remaining under 2.5 sec for all experiments. The verifier runtime remains constant under 5 msec. Later we will show that ZKP cost is a fraction of the total cost and add little overhead to the overall runtime.

\begin{figure}[t!]
\centerline{
\includegraphics[width = 0.9\columnwidth]{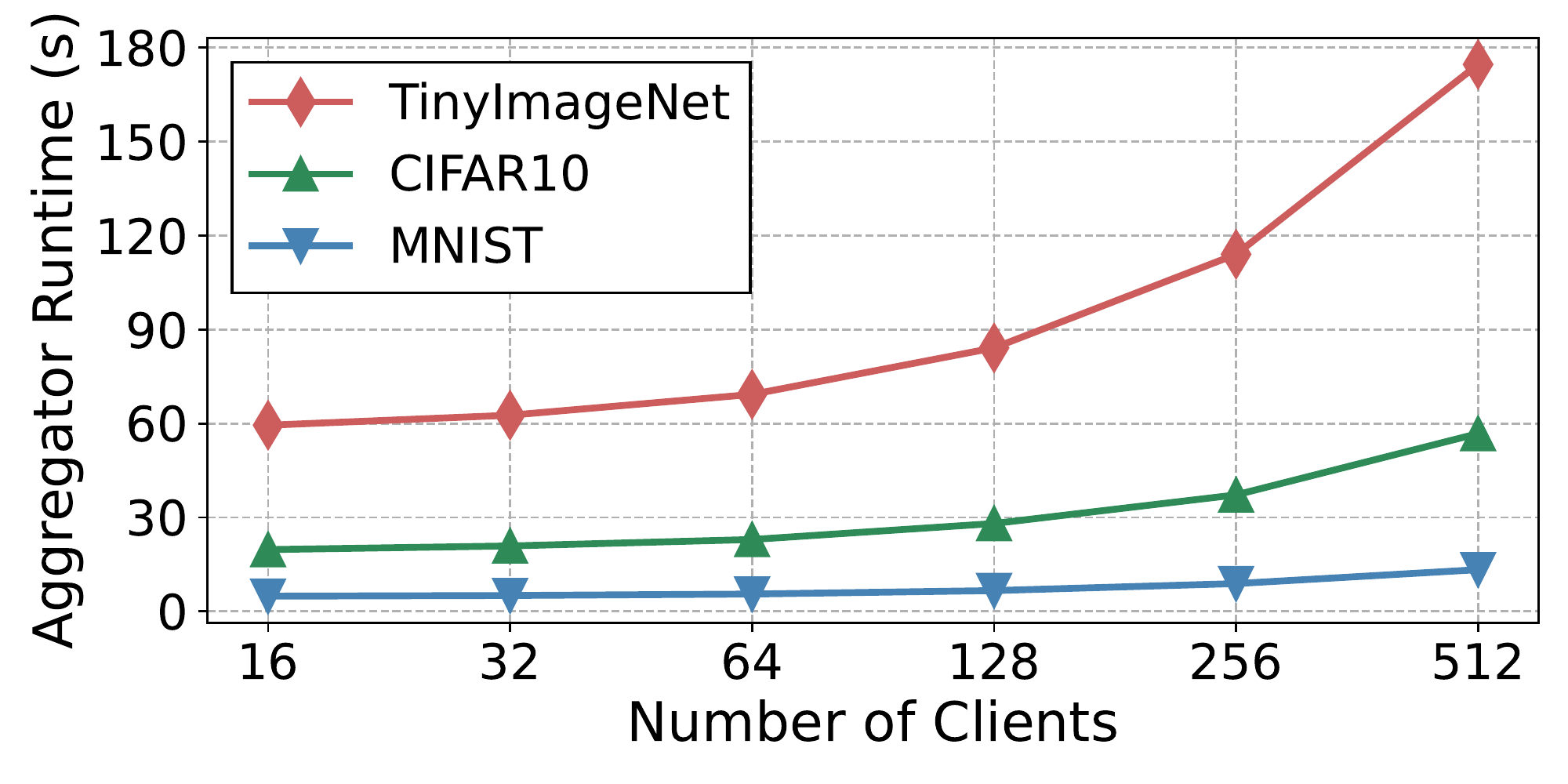}}
\vspace{-8pt}
\caption{Aggregator computation time for a training iteration for MNIST on LeNet5, CIFAR10 on ResNet20, and TinyImageNet on SqueezeNet.} 
\label{fig:agg-train}
\end{figure}

Next, we show results for an end-to-end training iteration. Aggregators' runtime during training is found in Figure~\ref{fig:agg-train} for different numbers of clients. Aggregators receive the ciphertexts from clients, perform aggregation using homomorphic addition over ciphertexts followed by a partial decryption after which the result is sent to the model owner. 
The decryption cost only depends on the size of the model, while the cost of aggregating ciphertexts increases with the number of clients. 
For our largest network of clients with 512 participants, the aggregator runtime is 13.3, 56.8, and 174.6 sec for MNIST, CIFAR10, and TinyImageNet benchmarks respectively.

\begin{figure}[t!]
\centering
\centerline{
\includegraphics[width = 0.9\columnwidth]{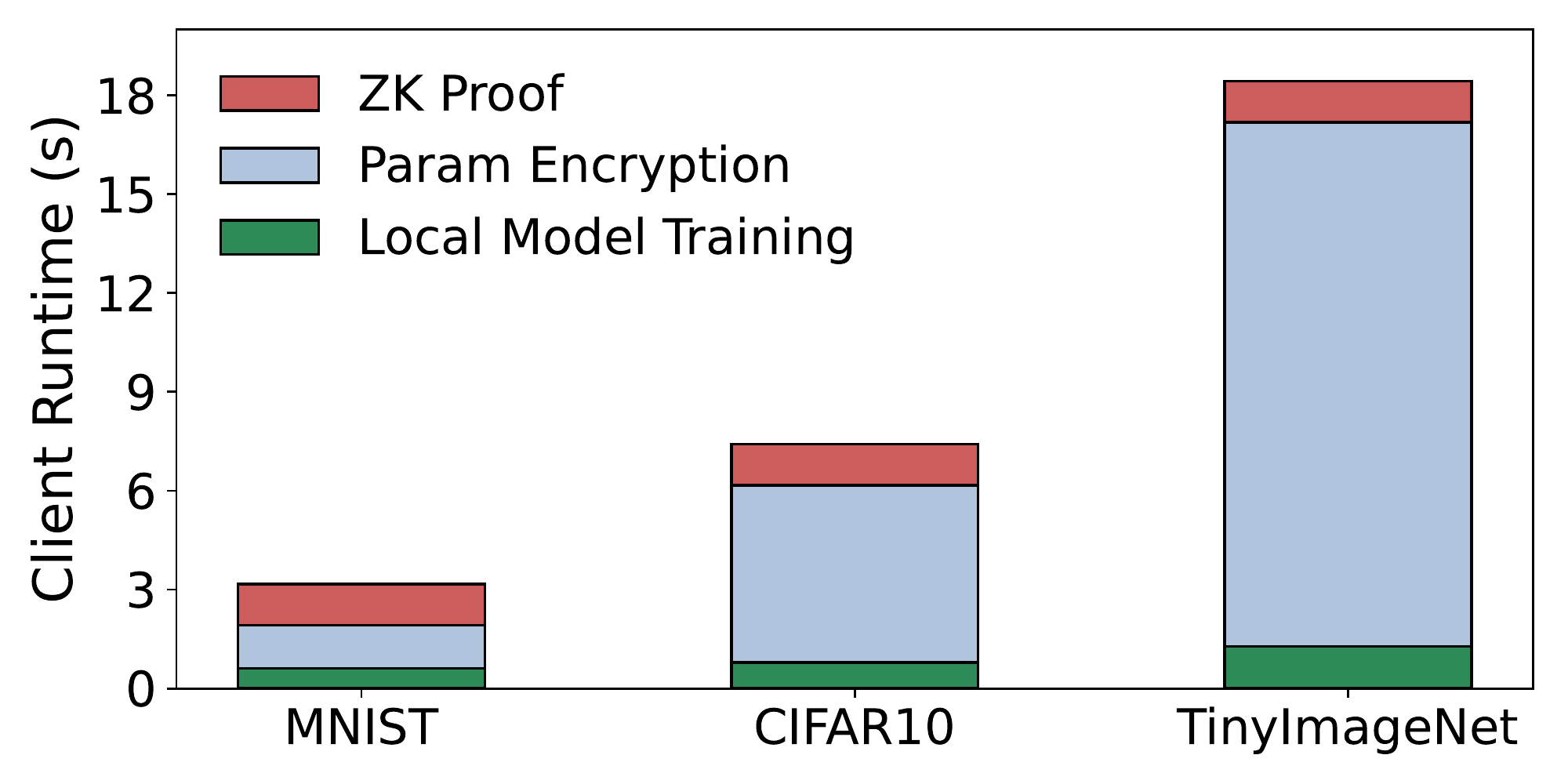}}
\vspace{-6pt}
\caption{Breakdown of client runtime for each training iteration for MNIST on LeNet5, CIFAR10 on ResNet20, and TinyImageNet on SqueezeNet} 
\label{fig:client-breakdown}
\end{figure}

\begin{figure*}[t!]
\centering
    \subfloat[]{{\includegraphics[width=0.32\textwidth]{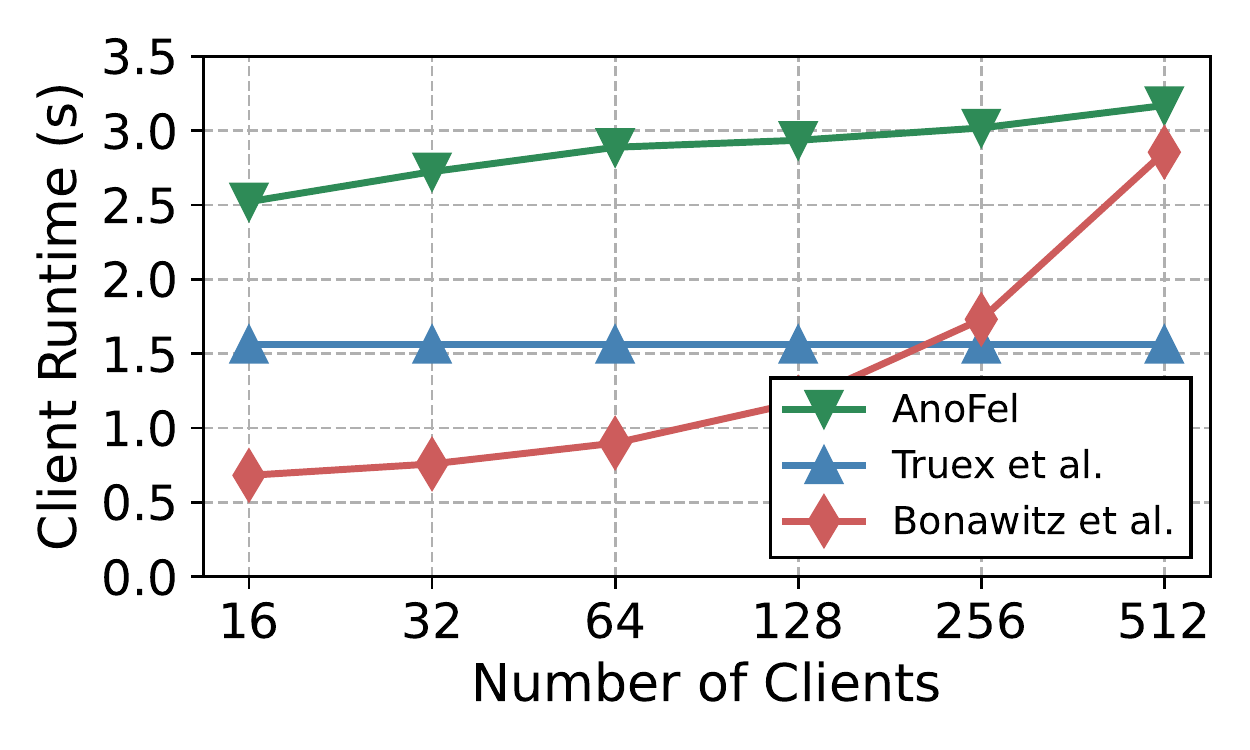} }} 
    \hfill
    \subfloat[]{{\includegraphics[width=0.32\textwidth]{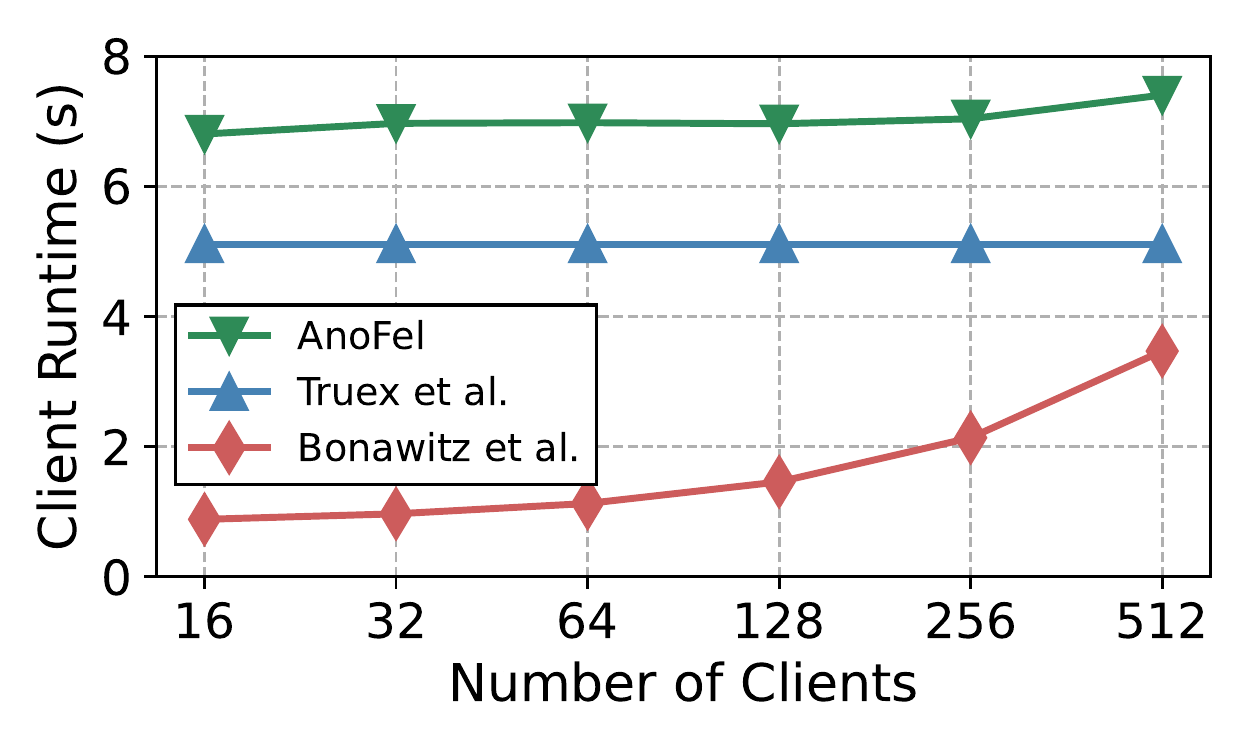} }}
    \hfill
    \subfloat[]{{\includegraphics[width=0.32\textwidth]{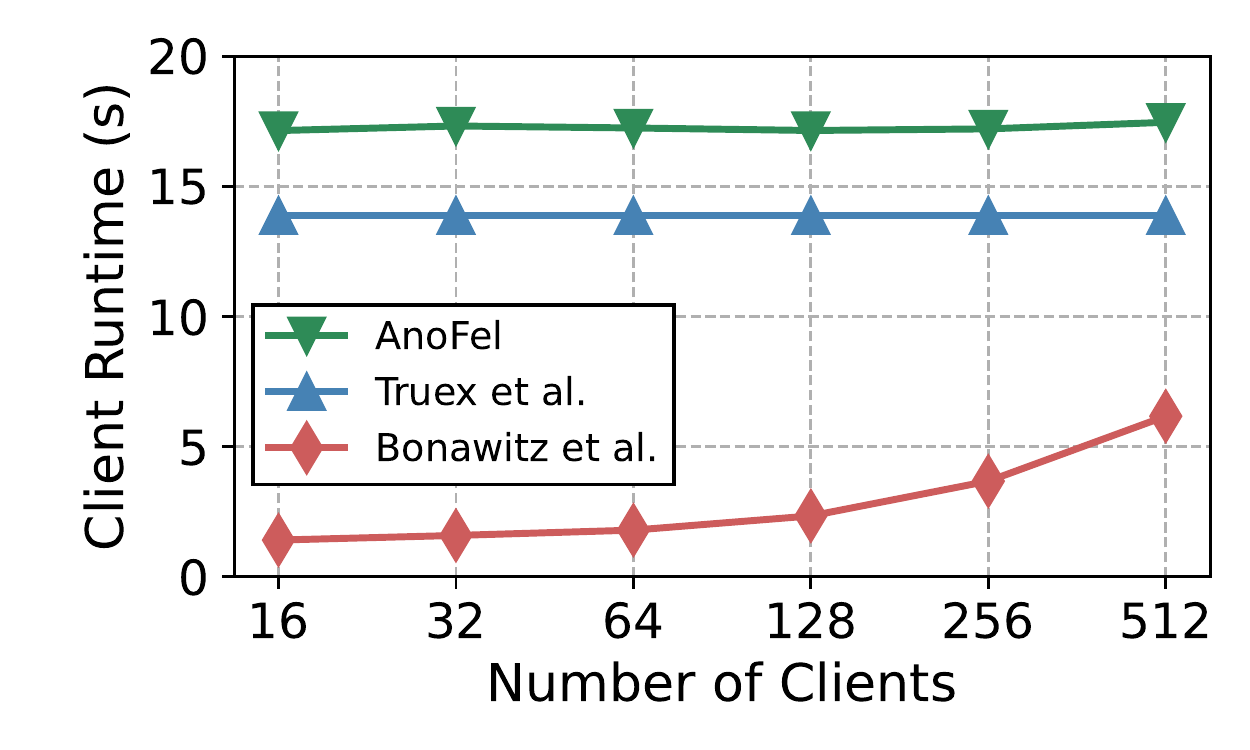} }}
\vspace{-12pt}
\caption{Comparing client computation time for one training round of \sysname, Truex et al.~\cite{truex2019hybrid}, and Bonawitz et al.~\cite{bonawitz2017practical} for (a) MNIST on LeNet5, (b) CIFAR10 on ResNet20, and (c) TinyImageNet on SqueezeNet. } 
\label{fig:client-train}
\end{figure*}

Figure~\ref{fig:client-breakdown} presents the client runtime with detailed breakdown for a training iteration with 512  clients for our benchmarks. The breakdown includes the cost of client local model training, model update encryption for the target $\agg$ aggregator committee, and ZKP generation. The cost of noising updates, generating keypairs, signatures, and hash computations are negligible, thus omitted from the plot.
Local model training cost is measured for training over a batch of 1024 images on a GeForce RTX 2080 Ti GPU, taking 0.616, 0.791, and 1.276 sec for MNIST, CIFAR10, and TinyImageNet benchmarks respectively. 
The remaining protocol costs (encryption and ZKP) are measured using the CPU described in configuration. 
For MNIST, the overhead of generating the ZKP is $~39\%$, and for the larger CIFAR10 and TinyImageNet benchmarks the ZKP overhead constitutes $17\%$ and $7\%$ of the total runtime. For all benchmarks, the runtime cost is dominated by the model update encryption step.

\vspace{4pt}
\noindent\textbf{Communication Overhead.} The communication overhead of the parties during setup and training phases are as follows.

\emph{Clients}: Client $i$'s setup involves the certifier's signature (96 B), and posting $m = (\comm_i, \taag, \pk_{sig}, \pi, \sid)$ and its signature (360 B). Each training iteration involves obtaining model parameters (121 KB, 527 KB, and 1.6 MB of 16-bit updates for LeNet5, ResNet20, and SqueezeNet respectively) and posting $m = (\mathbf{c}, AG_{\pk}, \taag, \pk_{sig}, \sid, \pi)$ and its signature (9.2 MB, 41 MB, and 124 MB for LeNet5, ResNet20, and SqueezeNet respectively).

\emph{Aggregators}: During setup, aggregators post their (signed) public key (160 B). In each training iteration, aggregators receive encrypted updates from clients and send partial decryptions to model owner (9.2 MB, 41 MB, and 124 MB for LeNet5, ResNet20, and SqueezeNet respectively).

\emph{Model Owner}: During each training iteration, model owner posts the model updates (124 KB and 546 KB for Lenet5 and ResNet20, respectively), and obtains partial ciphertexts from aggregators (9.2 MB, 41 MB, and 124 MB for LeNet5, ResNet20, and SqueezeNet respectively).

\vspace{4pt}
\noindent\textbf{Comparison to Baseline.} To better understand the performance of \sysname, we provide comparison to prior work on privacy-preserving federated learning. We don't know of any other framework that provides anonymity guarantees similar to \sysname, and therefore we chose two recent systems for privacy-preserving federated learning, namely, Truex et al.~\cite{truex2019hybrid} and Bonawitz et al.~\cite{bonawitz2017practical}. Truex et al. present a non-interactive protocol that deploys threshold homomorphic encryption for secure aggregation, and Bonawitz et al. develop an interactive protocol based on masking to protect the client updates during aggregation. 

We benchmarked the client runtime for a training iteration in Bonawitz et al. protocol based on implementation found at~\cite{google-impl} with fixes to allow more than 40 clients, and Truex et al. protocol using the implementation found at~\cite{truex-impl} for different number of clients. The results are shown in Figure~\ref{fig:client-train}.
When compared to Truex et al. protocol, \sysname is at most 2$\times$, 1.4$\times$, and 1.3$\times$ slower on MNIST, CIFAR10, and TinyImageNet respectively for different number of participants. 
Compared to Bonawitz et al. protocol, \sysname is at most 3.7$\times$, 7.7$\times$, and 12.1$\times$ slower on MNIST, CIFAR10, and TinyImageNet respectively. With larger number of clients, the runtime gap between Bonawitz et al. and \sysname reduces. For 512 participating clients \sysname is only slower by 1.1$\times$ on MNIST, 2.1$\times$ on CIFAR10, and 2.8$\times$ on TinyImageNet. Our results demonstrate the scalability of our framework; the cost of its additional anonymity guarantees, that are not supported by prior work, is relatively low especially in large scale scenarios.

\begin{figure*}[t!]
\centering
    \subfloat[]{{\includegraphics[width=0.24\textwidth]{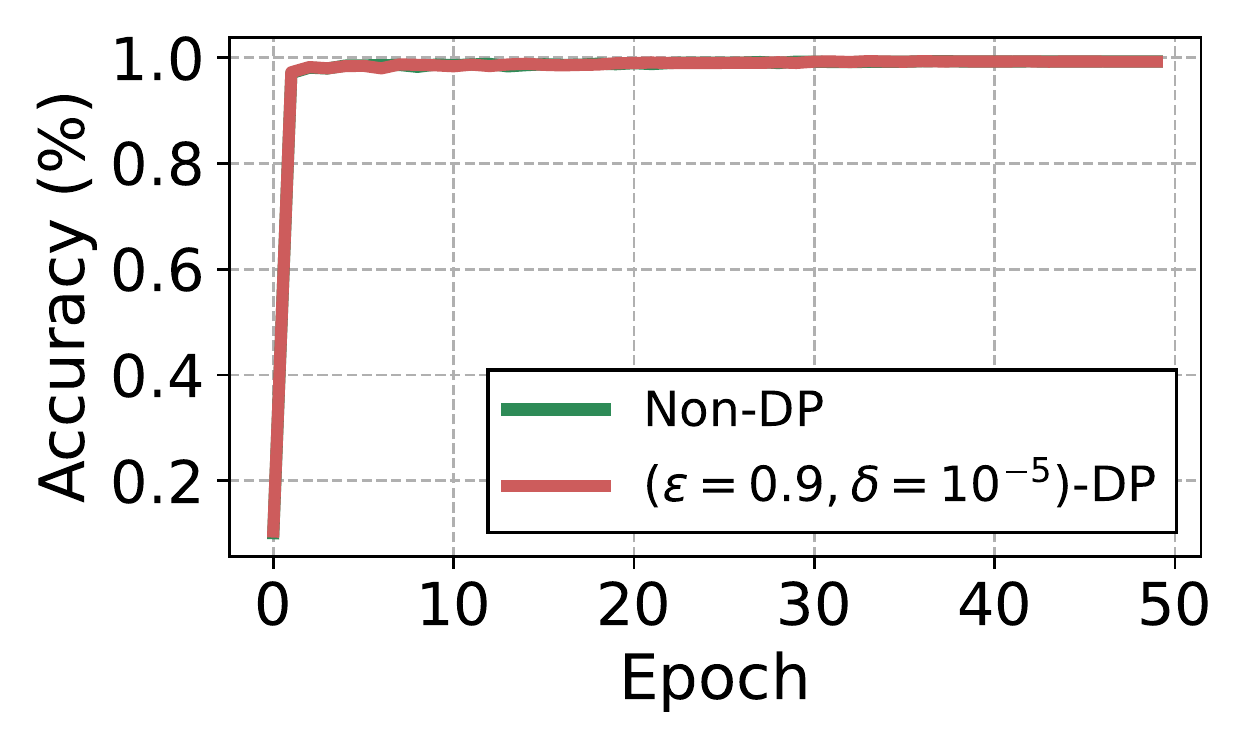} }} 
    \hfill
    \subfloat[]{{\includegraphics[width=0.24\textwidth]{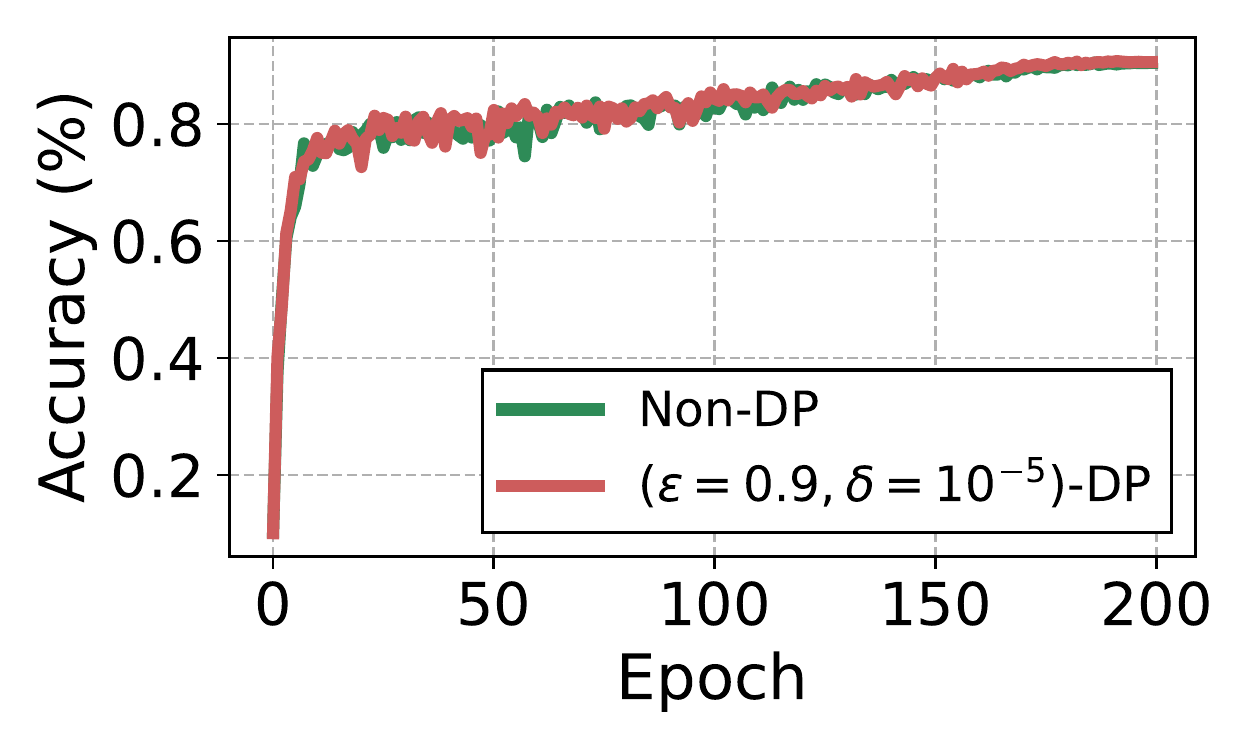} }}
    \hfill
    \subfloat[]{{\includegraphics[width=0.24\textwidth]{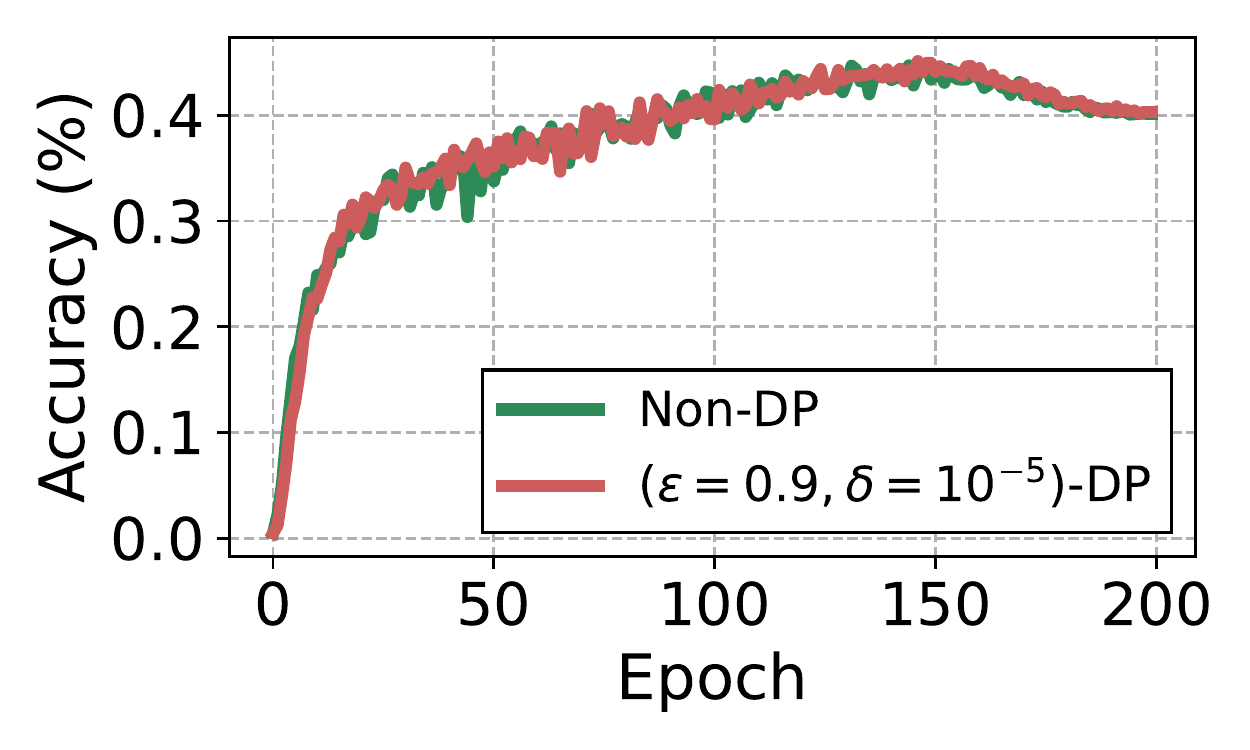} }}
    \hfill
    \subfloat[]{{\includegraphics[width=0.24\textwidth]{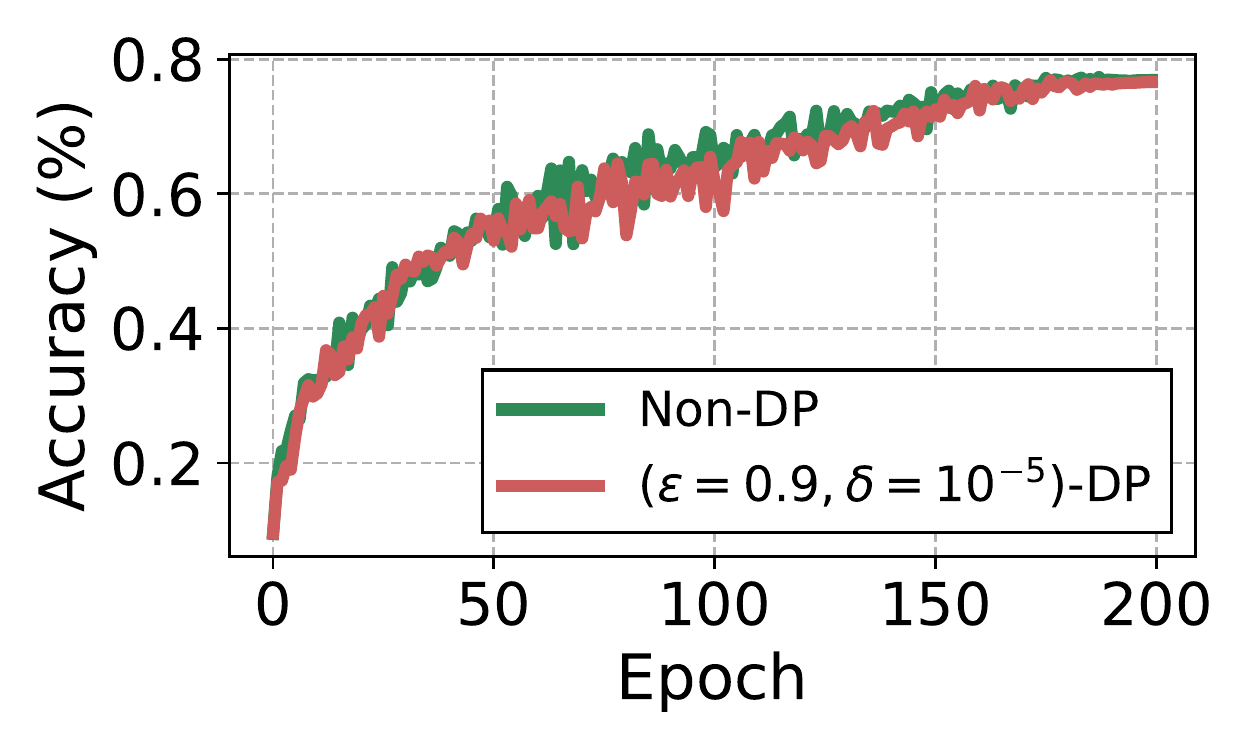} }}
\vspace{-12pt}
\caption{Accuracy of \sysname per epoch for $N=16$ clients compared with a non-DP baseline with IID datasets for (a) MNIST on LeNet5, (b) CIFAR10 on ResNet20, (c) TinyImageNet on SqueezeNet, and (d) non-IID dataset for CIFAR10 on ResNet20.} 
\label{fig:dp-iid}
\end{figure*}

\noindent\textbf{Model Accuracy.} We evaluate the model test accuracy in \sysname incorporating DP and compare to a non-DP baseline in Figure~\ref{fig:dp-iid}. For these experiments we assume number of clients $N=16$. Figure~\ref{fig:dp-iid}(a-c) presents the convergence characteristics of our benchmarks assuming IID datasets among clients. \sysname achieves $99.27\%$, $90.60\%$, and $40.30\%$ test accuracy on MNIST, CIFAR10, and TinyImageNet datasets respectively.
We also present accuracy results for non-IID setting, following the distribution described in~\cite{karimireddy2020byzantine} on CIFAR10. As depicted in Figure~\ref{fig:dp-iid}(d), \sysname achieves an accuracy of $76.61\%$. Across all benchmarks, \sysname obtains models with less than $0.5\%$ accuracy loss compared to non-DP models.

\section{Related Work}
\label{sec:related-work}
\vspace{-4pt}

\noindent{\bf Private federated learning.} Bonawitz et al.~\cite{bonawitz2017practical} is among the earliest works on private federated learning. It handles client dropouts (but not addition) using an interactive protocol. The proposed scheme does not support anonymity; each client is known by a logical identity, which in the malicious setting must be tied to a public key (through a PKI) to prevent impersonation. A followup work~\cite{bell2020secure} optimized the overhead of~\cite{bonawitz2017practical} and supported the semi-malicious model---the server is only trusted to handle client registration. This also violates anonymity since the server has full knowledge on the clients. The works~\cite{so2021turbo,kadhe2020fastsecagg,yang2021lightsecagg} also targeted the efficiency of~\cite{bonawitz2017practical}, and all inherit its lack of anonymity and interactivity issues.

Truex et al.~\cite{truex2019hybrid} use homomorphic encryption and differential privacy to achieve secure aggregation. This makes it easy to handle dropouts---since a user's update is independent of others'---but not addition since all users must be known at the setup phase to get shares of the decryption key. The proposed scheme relies on clients to decrypt the aggregated updates---which introduces excessive delays, and requires the server to know all clients and communicate with them directly---thus violating anonymity.  Xu et al.~\cite{xu2019hybridalpha} avoid the distributed decryption process, and allows for user additions (up to a maximum cap per iteration). However, this comes at the expense of introducing a trusted party to run the system setup and help in decrypting the aggregated model after knowing who participated in each iteration. Ryffel et al.~\cite{ryffel2020ariann} employ function secret sharing and assumes a fixed client participation, with a (semi-honest) server knows all clients and will communicate with them directly. Thus, client anonymity is not supported. Mo et al.~\cite{mo2021ppfl} use a trusted execution environment to achieve privacy. Beside not supporting anonymity, trusting a hardware is problematic due to the possibility of side channel and physical attacks.

\sysname addresses the limitations of prior work: it supports client anonymity, and does not involve them in the aggregation process. %
Accordingly, \sysname not only reduces overhead, but supports
dynamic participation without any additional recovery protocol or any trusted party. This is in addition to addressing recent attacks resulting from disseminating different initial models to clients~\cite{pasquini2021eluding} as explained previously.

\vspace{3pt}
\noindent{\bf Anonymity and Federated Learning.}
Several techniques were proposed to anonymize the dataset itself before using it in training, such as differential privacy~\cite{yeom2018privacy}, $k$-anonymity~\cite{sweeney2002k,choudhury2020anonymizing}, $l$-diversity~\cite{machanavajjhala2007diversity}, and $t$-closeness~\cite{li2007t} (a survey can be found in~\cite{el2022differential}). These techniques are considered complementary to \sysname: they allow anonymizing a dataset, and our system guarantees client's anonymity in the sense that no one will know if this client participated in training or which updates they have submitted.

The works~\cite{domingo2021secure,li2021privacy,hasirciouglu2021private,zhao2021anonymous,chen2022fedtor} target the same anonymity notion as in \sysname. %
Domingo et al. \cite{domingo2021secure} utilize probabilistic multi-hop routes for model update submission, with the clients known by fixed pseudonyms instead of their real identities. However, such pseudonyms provide only pseudoanonymity; several studies showed how network and traffic analysis can link these pseudonyms back to their real identities~\cite{reid13,koshy14,biryukov2019}. Also, their anonymity guarantees is based on assuming that clients do not collude with the model owner, a strong assumption that \sysname avoids. Li et al.~\cite{li2021privacy} uses interactive zero-knowledge proofs to achieve client anonymity when submitting model updates. Their approach, however, suffers from several security and technical issues: First, any party can generate a secret key and pass the proof challenge, not necessarily the intended client. Second, in this protocol, some parameters must be made public, but no details on how to do this in an anonymous way. Third, no discussion on how to preserve message integrity, making the protocol vulnerable to man-in-the-middle attacks.

The scheme proposed in~\cite{hasirciouglu2021private} works at the physical layer; it randomly samples a subset of clients' updates and aggregates their signals before submitting them to the model owner. The proposed protocol assumes clients are trusted, and does not discuss how to preserve integrity of the communicated updates. Zhao et al.~\cite{zhao2021anonymous} introduces a trust assumption to achieve anonymity; a trusted proxy server mediates communication between clients and model owner. Lastly, Chen et al.~\cite{chen2022fedtor} use a modified version of Tor to preserve anonymity; users authenticate each other and then negotiate symmetric keys to use for encryption. However, the negotiation and authentication processes are interactive, and the model owner records all clients' (chosen) identities, putting anonymity at risk due to the use of fixed identities. Thus, none of these systems supports anonymity in a provably secure way as \sysname does.

\section{Conclusion}
\label{sec:conclusions}
\vspace{-4pt}
In this paper, we presented \sysname, the first framework for private and anonymous user participation in federated learning. \sysname utilizes a public bulletin board, various cryptographic building blocks, and differential privacy to support dynamic and anonymous user participation, and secure aggregation. We also introduced the first formal security notion for private federated learning covering client anonymity. We demonstrated the efficiency and viability of \sysname through a concrete implementation and extensive benchmarking covering large scale scenarios and comparisons to prior work.

\section*{Acknowledgment}
The work of G.A. is supported by UConn's OVPR Research Excellence Program Award.

\bibliographystyle{plain}
\bibliography{anofelBib}

\begin{appendices}
\renewcommand{\thesection}{\Alph{section}}

\section{Proof of Theorem 1}
\label{app:proof}
To prove Theorem~\ref{theorem:pafl-sec}, we need to prove that \sysname does not impact training correctness, and that no $\ppt$ adversary can win the security games defined in Section~\ref{sec:security-def} for anonymity and dataset privacy with non-negligible probability.

Intuitively, \sysname satisfies these properties by relying on the correctness and security of the underlying cryptographic primitives, and the bounds on accuracy and privacy loss offered by DP as employed in our system. The use of a secure zero-knowledge proof guarantees: completeness (a valid honest proof generated by a client will be accepted by the bulletin board validators and the aggregators $\agg$), soundness (a client that does not own a certified dataset cannot register, and a client that does not belong to the registered set cannot forge valid proofs during training), and zero-knowledge (so the proof does not reveal anything about the master public key of the client or its dataset). 

Furthermore, the use of a semantically secure threshold homomorphic encryption scheme guarantees that the ciphertexts of the model updates do not reveal anything about the underlying (plaintext) updates, and adding them will produce a valid result of the sum of these updates. The use of a secure commitment scheme guarantees that a commitment posted by a client $\cl_i$ hides the dataset $D_i$ and binds this client to $D_i$. The security of the digital signature scheme and the PKI guarantees that a malicious adversary cannot forge a certificate for a corrupted dataset, and that a man-in-the-middle attacker cannot manipulate any of the messages that a client, aggregator, or a server send. Also, under the assumption that at least $t$ members of $\agg$ are honest, this guarantees that $S$ will not have access to the individual updates submitted by clients.

Moreover, the use of a $(\epsilon,\delta)$-differential privacy technique leads to small advantage for the attacker in membership attacks as given by equation~\ref{eq:gamma}, and a small error (or loss in accuracy) bound as detailed in Section~\ref{sec:prelim}. We use these parameters in our proofs below.

Formally, the proof of Theorem~\ref{theorem:pafl-sec} requires proving three lemmas showing that \sysname is correct, anonymous, and supports dataset privacy. For correctness, we remark that \sysname does not impact training correctness and accuracy. So if a defense mechanism against inference attacks is employed, and this mechanism provides a trade-off with respect to accuracy (as in differential privacy), \sysname will not impact that level. 

\begin{lemma}\label{lemma:correctness}
\sysname satisfies the correctness property as defined in Definition~\ref{def:pafl}.
\end{lemma}

\begin{proof}
Correctness follows by the correctness of the homomorphic encryption scheme and the security of the digital signature, as well as the accuracy level provided by DP. A semi-honest client in the training phase will perform training as required and encrypt the updates and post them on the board. Since \sysname uses an existential unforgeable digital signature scheme, a malicious attacker $\adv$ cannot modify the ciphertext of the updates without invalidating the signature, and $\adv$ cannot forge a valid signature over a modified ciphertext. Thus, it is guaranteed that all accepted model updates ciphertexts are the ones produced by the client. Also, since \sysname uses a correct (and secure) homomorphic encryption scheme, the homomorphic addition of the ciphertexts will produce a ciphertext of the sum of the actual updates (with their added noise level by DP). By the correctness of the decryption algorithm, after decrypting the sum ciphertext, the server will obtain the correct value of the aggregated updates in each training iteration. This trained model differs from the actual one by the error bounds $\alpha$ obtained from DP, thus satisfying $\alpha$-correctness. This completes the proof.
\end{proof}

For anonymity, as we mentioned previously, inference attacks will have no impact on anonymity unless the leaked datapoint contains sensitive data (like the identity of the client)---so this attack assumes that the adversary got a hold on the dataset or part of it. As an extra step, a technique can be used to pre-process the dataset to remove sensitive attributes from the dataset, satisfies that and renders membership attacks ineffective in compromising anonymity (note if the attacker gets a hold on a client dataset and he knows the client identity, then he already compromised privacy of that client). The definition of our anonymity property does not assume the adversary knows the dataset or identity of the clients involved in the challenge.\footnote{Even if we allow that, we can define $\gamma$-anonymity property where the attacker wins with probability bounded by $\gamma$ inherited from DP.}

\begin{lemma}\label{lemma:anonymity}
\sysname satisfies the anonymity property as defined in Definition~\ref{def:pafl}.
\end{lemma}
\begin{proof}
Under the assumption that at least one honest client (other than $\cl$) has submitted updates during the challenge training iteration (as described in the game definition earlier), accessing the aggregated model updates at the end of any iteration will not provide $\adv$ with any non-negligible advantage in winning the game. Thus, the proof is reduced to showing that all actions introduced by \sysname preserve anonymity. We prove that using a similar proof technique to the one in~\cite{quisquis}, where we show a series of hybrids starting with an $\anongame$ with $b = 0$ ($\hybrid_0$), and finishing with an $\anongame$ game with $b = 1$ ($\hybrid_7$). By showing that all these hybrids are indistinguishable, this proves that $\adv$ cannot tell which client was chosen for the challenge $\train$ command in $\anongame$. Now, we proceed with a sequence of hybrid games as follows: \medskip

\noindent\underline{$\hybrid_0$}: The game $\anongame$ with $b = 0$. \medskip

\noindent\underline{$\hybrid_1$}: Same as $\hybrid_0$, but we replace the zero-knowledge proofs with simulated ones, i.e., we invoke the zero-knowledge property simulator for each of the $\register$ and $\train$ queries, and we replace the actual proofs in the output of these queries with simulated ones. The hybrids $\hybrid_0$ and $\hybrid_1$ are indistinguishable by the zero-knowledge property of the ZKP system that \sysname uses. That is, if $\adv$ can distinguish them, then we can build another adversary $\adv'$ that can break the zero-knowledge property, which is a contradiction. \medskip

\noindent\underline{$\hybrid_2$}: Same as $\hybrid_1$, but we replace $(\cl_0, \aux_0, \cl_1, \aux_1)$ with fresh output $(\cl_0', \aux_0', \cl_1', \aux_1')$. That is, we choose fresh datasets and register two fresh clients using them. So if $\setup$ created a state with $n$ clients, any $\register$ query for any of the $n$ clients other than $\cl_0$ and $\cl_1$ will proceed as in $\hybrid_1$. However, if it is for $\cl_0$ or $\cl_1$, we replace them with $\cl_0'$ or $\cl_1'$ and proceed.

The hybrids $\hybrid_1$ and $\hybrid_2$ are indistinguishable by the zero-knowledge property of the ZKP system and the hiding property of the commitment scheme that \sysname uses (which implies that client registration is indistinguishable). That is, if $\adv$ can distinguish them, then we can build two adversaries: $\adv'$ that can break the zero-knowledge property of the ZKP system, and $\adv''$ that can break the hiding property of the commitment scheme, which is a contradiction. \medskip

\noindent\underline{$\hybrid_3$}: Same as $\hybrid_2$, but we replace the output of training any of $(\cl_0, \aux_0, \cl_1, \aux_1)$ with fresh output produced by training $(\cl_0', \aux_0', \cl_1', \aux_1')$. As above, if $\setup$ created a state with $n$ client registrations, any $\train$ query for any of the $n$ clients other than $\cl_0$ and $\cl_1$ will proceed as in $\hybrid_2$. However, if the train query is for $\cl_0$ or $\cl_1$, we replace them with training output based on the fresh datasets owned by $\cl_0'$ or $\cl_1'$ and proceed.

$\hybrid_2$ and $\hybrid_3$ are indistinguishable by the zero knowledge property of the ZKP system and the semantic security of the homomorphic encryption scheme used in \sysname (which implies that training is indistinguishable). If $\adv$ can distinguish them, then we can build two adversaries: $\adv'$ that can break the zero-knowledge property of the ZKP system, and $\adv''$ that can break the semantic security of the encryption scheme, which is a contradiction. \medskip

\noindent\underline{$\hybrid_4$}: Same as $\hybrid_3$, but with $b = 1$. The hybrids $\hybrid_3$ and $\hybrid_4$ are indistinguishable by the indistinguishability of model training as described above. \medskip

\noindent\underline{$\hybrid_5$}: Same as $\hybrid_4$, but with $(\cl_0, \aux_0, \cl_1, \aux_1)$ used in training as in the original game. So this is $\hybrid_3$ with $b = 1$. The hybrids $\hybrid_4$ and $\hybrid_5$ are indistinguishable by the indistinguishability argument of $\hybrid_3$ and $\hybrid_2$. \medskip

\noindent\underline{$\hybrid_6$}: Same as $\hybrid_5$, but with $(\cl_0, \aux_0, \cl_1, \aux_1)$ used in registration as in the original game. So this is $\hybrid_2$ with $b = 1$. The hybrids $\hybrid_5$ and $\hybrid_6$ are indistinguishable by the indistinguishability argument of $\hybrid_2$ and $\hybrid_1$. \medskip

\noindent\underline{$\hybrid_7$}: Same as $\hybrid_6$, but with real ZKPs instead of the simulated ones. So this is the original $\anongame$ with $b = 1$. The hybrids $\hybrid_6$ and $\hybrid_7$ are indistinguishable by the indistinguishability argument of $\hybrid_1$ and $\hybrid_0$. \medskip

This shows that $\anongame$ with $b = 0$ is indistinguishable from $\anongame$ with $b = 1$, which completes the proof.
\end{proof}

As for dataset privacy, membership attacks will allow $\adv$ to win the dataset privacy game with advantage $\gamma$ as he selects the datasets involved in the challenge. Thus, our proofs proceeds in two stages: first, we show that the cryptographic pritmitives used in \sysname do not provide $\adv$ with any non-negligible advantage, and second, by the security guarantees of DP, this attacker has an advantage bounded by $\gamma$ due to membership attacks.

\begin{lemma}\label{lemma:privacy}
\sysname satisfies the dataset privacy property as defined in Definition~\ref{def:pafl}.
\end{lemma}
\begin{proof}
As defined in the $\dindgame$, adversary $\adv$ chooses two datasets $D_0$ and $D_1$. Then, the challenger picks one of these datasets at random, registers a client with that dataset, and invokes the $\train$ command for that client. $\adv$ gets to see the output of the registration and training commands, which are the messages and signatures that a client sends in the setup and training phases of \sysname as described before.

In order to win the $\dindgame$, $\adv$ can attack the registration or the training process. That is, for the former $\adv$ tries to reveal which dataset is hidden in the posted commitment or obtain information about the witness underlying the submitted proof, which contains the dataset $D_b$ in this case. While for the latter, $\adv$ may try to infer any information about the plaintext of the model updates ciphertext (recall that the model parameters in a training iteration are public, and thus, $\adv$ can produce the model updates resulted from the use of $D_0$ and $D_1$). Note that attacking the ZKP to reveal any information about the commitment that was used, and then attacking that commitment, reduces to the same case of attacking the registration process.

Attacking registration means attempting to break the hiding property of the commitment scheme and the zero-knowledge property of the ZKP system. Since \sysname uses a secure commitment scheme, the former will succeed with negligible probability $\negl_1(\lambda)$. Also, since \sysname uses a secure ZKP system that satisfies the zero-knowledge property, the latter will succeed with negligible probability $\negl_2(\lambda)$. Attacking the training process means attempting to break the semantic security of the encryption scheme. Since \sysname uses a semantically secure encryption to encrypt the model updates, such an attack will succeed with negligible probability $\negl_3(\lambda)$.

Accordingly, $\adv$'s advantage by the cryptographic primitives that we use is $\negl_1(\lambda) + \negl_3(\lambda) + \negl_3(\lambda) = \negl(\lambda)$.

Now, $\adv$ can query the oracle $\pafl$ to access the updated model and perform a membership attack. That is, $\adv$ knows both datasets and query the model over various datapoints to see which dataset was used in training. The success of this strategy is bounded by the privacy loss $\gamma$ given by equation~\ref{eq:gamma}. 

Thus, the probability that $\adv$ wins in the $\dindgame$ is $\frac{1}{2}+ \negl(\lambda) + \gamma$, which completes the proof.
\end{proof}

\noindent{\bf Proof of Theorem~\ref{theorem:pafl-sec}.} Follows by Lemmas~\ref{lemma:correctness},~\ref{lemma:anonymity}, and~\ref{lemma:privacy}.

\end{appendices}

\end{document}